\documentclass[12pt]{article}
\usepackage{latexsym,amsmath,amscd,amssymb,graphics,color}
\usepackage{amsthm,url,bm}
\usepackage{graphicx}
\usepackage{epsfig, picinpar,framed}
\usepackage[margin=2cm]{geometry}

\newcommand{\bfi}{\bfseries\itshape}

\newcommand{\rem}[1]{}

\def\b{\begin{eqnarray}}
\def\e{\end{eqnarray}}

\def\tfrac#1{\frac{ #1}}
\def\Ad{{\rm Ad}}
\def\ad{{\rm ad}}

\newtheorem{theorem}{Theorem}[section]
\newtheorem{definition}[theorem]{Definition}
\newtheorem{lemma}[theorem]{Lemma}
\newtheorem{remark}[theorem]{Remark}
\newtheorem{proposition}[theorem]{Proposition}

\makeatletter\@addtoreset{equation}{section}\makeatother

\begin{document}


\title{$G$-Strands}
\author{Darryl D. Holm$^{1}$, Rossen I. Ivanov$^{2}$, James R. Percival$^{1}$}
\addtocounter{footnote}{1}
\footnotetext{Department of Mathematics, Imperial College London. London SW7 2AZ, UK.
\texttt{d.holm@ic.ac.uk, j.percival@ic.ac.uk}
\addtocounter{footnote}{1}}
\footnotetext{School of Mathematical Sciences, Dublin Institute of Technology, Kevin Street, Dublin 8, Ireland,  \texttt{rivanov@dit.ie}
}

\date{\it Fondly remembering our late friend, Jerry Marsden}
\maketitle

\makeatother

\maketitle

\rem{
PACS numbers:
}



\begin{abstract}

A $G$-strand is a map $g(t,{s}):\,\mathbb{R}\times\mathbb{R}\to G$ for a Lie group $G$ that follows from Hamilton's principle for a certain class of $G$-invariant Lagrangians. The $SO(3)$-strand is the $G$-strand version of the rigid body equation and it may be regarded physically as a continuous spin chain. Here, $SO(3)_K$-strand dynamics for ellipsoidal rotations is derived as an Euler-Poincar\'e system for a certain class of variations and recast as a Lie-Poisson system for coadjoint flow with the same Hamiltonian structure as for a perfect complex fluid. For a special Hamiltonian, the $SO(3)_K$-strand is mapped into a completely integrable generalization of the classical chiral model for the $SO(3)$-strand. Analogous results are obtained for the $Sp(2)$-strand. The $Sp(2)$-strand is the $G$-strand version of the $Sp(2)$ Bloch-Iserles ordinary differential equation, whose solutions exhibit dynamical sorting. Numerical solutions show nonlinear interactions of coherent wave-like solutions in both cases. ${\rm Diff}(\mathbb{R})$-strand equations on the diffeomorphism group $G={\rm Diff}(\mathbb{R})$ are also introduced  and shown to admit solutions with singular support (e.g., peakons).
\end{abstract}

\newpage

\tableofcontents

\section{Euler-Poincar\'e equations for a $G$-strand}

The Euler-Poincar\'e (EP) theory of $G$-strands is an extension of the classical chiral models. The classical chiral models are nonlinear relativistically invariant Lagrangian field theories on group manifolds. As such, they are fundamental in theoretical physics. The vast literature of results for these models is fascinating. For example, it is well known that these models are integrable in 1 + 1 dimensions and possess soliton solutions. 
See \cite{Ch1981,Po1976,ZaMi1978,ZaMi1980,deVega1979,NoMaPiZa1984,Mai86,{Man1994}} and references therein for discussions of the many aspects of integrability of the chiral models, including the famous dressing method for explicitly deriving the soliton solutions of these models, which is given in \cite{ZaMi1980}. The solitons for the $O(3)$ chiral model are particularly familiar, because this model allows an integrable reduction to the one-component sine-Gordon equation (e.g. see \cite{NoMaPiZa1984}). 
Many generalizations of these models have been introduced. For example, generalized chiral models with metrics that are not ad-invariant on the Lie algebra are considered in \cite{So1997}. Other generalizations of chiral models for non-semisimple groups are studied in \cite{HlaSno2001}. An integrable chiral model in 2+1 dimensions was proposed in \cite{Wa1988}. Finally, the possibility for fermionic interpretatation of the current variables was explored in \cite{Wi1984}.

\paragraph{Left $G$-Invariant Lagrangian.}
We begin with the following ingredients of EP theory. For more details and in-depth discussion, see \cite{Ho2011,HoMaRa1998}.
\begin{itemize}
\item
Let $G$ be a Lie group. 
A map $g(t,s): \mathbb{R}\times \mathbb{R}\to G$ has two types of tangent vectors, $\dot{g} := g_t \in u_g$
and $g' :=g_s \in v_g$
\item Assume that the function $ L(g,u_g,v_g) : T G\times TG \rightarrow \mathbb{R}$ is left $G$-invariant.

\item  Left $G$--invariance of $L$ permits us to define
$l: \mathfrak{g}\times \mathfrak{g} \rightarrow \mathbb{R}$ by
\[
l(g^{-1} u_g , g^{-1} v_g ) = L(g,u_g,v_g).
\]
Conversely,  this relation defines for any
$l: \mathfrak{g}\times \mathfrak{g} \rightarrow \mathbb{R} $ a left $G$-invariant function
$ L : T G\times TG \rightarrow \mathbb{R} $.
\item For a map $g(t,s): \mathbb{R}\times \mathbb{R}\to G$ let
\[ 
\mathsf{X} (t,s) := g^{ -1} g_t (t,s) =g^{ -1}\dot{g}(t,s)
\quad\hbox{and}\quad
\mathsf{Y} (t,s) := g^{ -1} g_s (t,s)= g^{ -1} g' (t,s)
.\]
\end{itemize}
\begin{lemma}
The left-trivialized tangent vectors $\mathsf{X} (t,s)$ and $\mathsf{Y} (t,s)$ at the identity of $G$ satisfy
\begin{equation} 
\mathsf{Y}_t - \mathsf{X}_s = -\,{\rm ad}_\mathsf{X}\mathsf{Y}
\,.
\label{zero-curv}
\end{equation} 
\end{lemma}
\begin{proof}
The proof is standard and follows from equality of cross derivatives $g_{ts}=g_{st}$, cf. \cite{Ho2011,HoMaRa1998}.
As a consequence, equation (\ref{zero-curv}) is often called a \emph{zero-curvature relation}.
\end{proof}
\begin{theorem} [Euler-Poincar\'e theorem]\label{lall}$\,$

With the preceding notation, the following statements are equivalent:
\begin{enumerate}
\item [{\bf i} ] Hamilton's variational principle on $T G\times TG$
\begin{equation} \label{hamiltonprinciple}
\delta \int _{t_1} ^{t_2} L(g(t,s), \dot{g} (t,s), g'(t,s) ) \,ds\,dt = 0
\end{equation}
holds, for variations $\delta g(t,s)$
of $ g (t,s) $ vanishing at the endpoints in $t$ and $s$.
\item [{\bf ii}  ] The function $g(t,s)$ satisfies Euler--Lagrange
equations for $L$ on $G$, given by
\begin{equation} \label{EL-eqns}
\frac{\partial L}{\partial g} - \frac{\partial}{\partial t}\frac{\partial L}{\partial g_t}
- \frac{\partial}{\partial s}\frac{\partial L}{\partial g_s} = 0
\end{equation}

\item [{\bf iii} ]  The constrained variational principle%
\footnote{As with the basic Euler--Poincar\'e equations \cite{MaRa1999}, this is not
strictly a variational principle in the
same sense as the standard Hamilton's principle.
It is more like the Lagrange d'Alembert
principle, because we impose the stated constraints
on the variations allowed \cite{HoMaRa1998}.}
\begin{equation} \label{variationalprinciple}
\delta \int _{t_1} ^{t_2}  l(\mathsf{X}(t,s), \mathsf{Y}(t,s)) \,ds\,dt = 0
\end{equation}
holds on $\mathfrak{g}\times\mathfrak{g}$, using variations of $ \mathsf{X} $ and $\mathsf{Y}$ of the forms
\begin{equation} \label{epvariations}
\delta \mathsf{X} = \dot{\mathsf{Z} } + {\rm ad}_\mathsf{X}\mathsf{Z}
\quad\hbox{and}\quad
\delta \mathsf{Y} = \mathsf{Z}\,' + {\rm ad}_\mathsf{Y} \mathsf{Z} \,, 
\end{equation}
where $\mathsf{Z}(t,s) \in \mathfrak{g}$ vanishes at the endpoints.
\item [{\bf iv}] The {\bfi Euler--Poincar\'{e}}
equations hold on $\mathfrak{g}^*\times\mathfrak{g}^*$
\begin{equation}
\frac{d}{dt} \frac{\delta l}{\delta \mathsf{X}} -
 \operatorname{ad}_{\mathsf{X}}^{\ast} \frac{ \delta l }{ \delta \mathsf{X}}
+
\frac{d}{ds} \frac{\delta l}{\delta \mathsf{Y}} 
-
 \operatorname{ad}_{\mathsf{Y}}^{\ast} \frac{ \delta l }{ \delta \mathsf{Y}}
 =
 0
\, .
\label{EPst-eqn}
\end{equation}
\end{enumerate}
\end{theorem}

\begin{proof}
The equivalence of {\bf i} and {\bf ii}
holds for any configuration manifold and so, in particular, it
holds in this case.

Next we show the equivalence of {\bf iii} and {\bf iv}.
Indeed, using the definitions,
integrating by parts, and taking into
account that $\mathsf{Z}$ vanishes at the endpoints $(t_1,s_1,t_2,s_2)$, allows one to compute the
variation of the integral as
\begin{align}
\begin{split}
\delta \int_{t_1}^{t_2}\!\!\!\! \int_{s_1}^{s_2}
 l(\mathsf{X}(t,s), \mathsf{Y}(t,s)) \,ds\,dt 
 &=
\int_{t_1}^{t_2}\!\!\!\! \int_{s_1}^{s_2}
\left \langle
\frac{ \delta l}{\delta \mathsf{X}}\,,
\delta \mathsf{X} 
\right \rangle 
+ \left \langle 
\frac{ \delta l}{\delta \mathsf{Y}}\,,
\delta \mathsf{Y} 
\right \rangle \,ds\, dt
\\  &=
\int_{t_1}^{t_2}\!\!\!\! \int_{s_1}^{s_2}
\left \langle
\frac{\delta l}{\delta \mathsf{X}}\,, \dot{\mathsf{Z}} + \operatorname{ad}_{\mathsf{X}} \mathsf{Z} 
 \right \rangle
 +
 \left \langle
\frac{\delta l}{\delta \mathsf{Y}}\,, \mathsf{Z}\,'  +  \operatorname{ad}_{\mathsf{Y}} \mathsf{Z} 
 \right \rangle
 \,ds\, dt
\\ &=
- \int_{t_1}^{t_2}\!\!\!\! \int_{s_1}^{s_2}
 \left \langle 
\frac{ d}{dt}
\frac{\delta l }{ \delta \mathsf{X}}  -
 \operatorname{ad}_{\mathsf{X}}^{\ast}\frac{\delta l}{\delta \mathsf{X}}
 +
 \frac{ d}{ds}
 \frac{\delta l }{ \delta \mathsf{Y}}  -
 \operatorname{ad}_{\mathsf{Y}}^{\ast}\frac{\delta l}{\delta \mathsf{Y}}
 \,, \mathsf{Z}
\right \rangle 
\,ds\, dt
\,.
\end{split}
\label{EPst-calc}
 \end{align}
To understand the notation in the last step, recall that the coadjoint action $(\ad^*: \mathfrak{g}\times \mathfrak{g}^*\to \mathfrak{g}^* )$ is defined as the dual of the adjoint action $(\ad: \mathfrak{g}\times \mathfrak{g}\to \mathfrak{g})$ with respect to the $L^2$ pairing $\langle\,\cdot\,,\,\cdot\,\rangle: \mathfrak{g}^* \times \mathfrak{g}\to \mathbb{R}$ by
\begin{equation}
\Big\langle
{\rm ad}^*_{\mathsf{X}}\frac{\delta\ell}{\delta {\mathsf{X}}}
\,,\,{\mathsf{Z}} \Big\rangle
=
\Big\langle \frac{\delta\ell}{\delta {\mathsf{X}}}
\,,\,{\rm ad}_{\mathsf{X}}{\mathsf{Z}} \Big\rangle
\,.
\label{ad-star-def}
\end{equation}
Thus, the calculation in (\ref{EPst-calc}) allows us to conclude that {\bf iii} and {\bf iv} are equivalent.

Finally we show that {\bf i} and {\bf iii} are equivalent.
First note that the left $G$-invariance of $L:TG
\times TG \rightarrow \mathbb{R}$ and the definitions of $\mathsf{X}(t,s)$ 
and $\mathsf{Y}(t,s)$ imply that the
integrands in (\ref{hamiltonprinciple}) and
(\ref{variationalprinciple}) are equal. Moreover, all variations
$\delta g(t,s) \in TG\times TG$ of $g(t,s)$ with fixed endpoints induce and are
induced by variations $(\delta \mathsf{X}(t,s),\delta \mathsf{Y}(t,s)) \in \mathfrak{g}\times \mathfrak{g}$ 
of $\mathsf{X}(t,s)$ and $\mathsf{Y}(t,s)$ of the form 
\begin{equation}
\delta \mathsf{X} = \dot{\mathsf{Z} } + {\rm ad}_\mathsf{X}\mathsf{Z}
\quad\hbox{and}\quad
\delta \mathsf{Y} = \mathsf{Z}\,' + {\rm ad}_\mathsf{Y} \mathsf{Z} \,, 
\label{delta-xi-zeta}
\end{equation}
with $\mathsf{Z}(t,s) \in \mathfrak{g}$ vanishing at the
endpoints. The relation between $\delta g(t,s)$ and $\mathsf{Z}(t,s)$
is given by $\mathsf{Z}(t,s) = g(t,s)^{-1}\delta g(t,s) $. 

Thus, if {\bf i} holds, we define $\mathsf{Z}(t,s) =  g^{-1}\delta  g(t,s)$ 
for a variation $\delta g(t,s)$ with fixed endpoints. Then
if we let $\mathsf{X} = g^{-1}\dot{g} (t,s)$ and $\mathsf{Y} = g^{-1}g' (t,s)$, we have 
equation (\ref{delta-xi-zeta}) by the same standard calculation that produced (\ref{zero-curv}).  
Conversely, if  equation (\ref{delta-xi-zeta}) holds with $\mathsf{Z}(t,s)$ vanishing at the
endpoints, we define $\delta g(t,s) =g(t,s) \circ  \mathsf{Z}(t,s)$ and the above
proposition guarantees then that this $\delta g(t,s)$ is the
general variation of $g(t,s)$ vanishing at the endpoints. 
Hence, {\bf i} and {\bf iii} are equivalent.
\end{proof}

\begin{remark}
For a \emph{right} $G$-invariant Lagrangian, the results and proofs are the same as above, except for
the \emph{sign change} appearing in the variations, for which $\delta \mathsf{X} = \dot{\mathsf{Z} } - {\rm ad}_\mathsf{X}\mathsf{Z}$ 
and $\delta \mathsf{Y} = \mathsf{Z}\,' - {\rm ad}_\mathsf{Y} \mathsf{Z}$.
\end{remark}

\paragraph{Evolutionary $G$-strand.}
We now define the fundamental quantity of interest in the remainder of the paper. 

\begin{definition}
A {\bf $G$-strand} is an evolutionary map  into a Lie group $G$, $g(t,{s}):\,\mathbb{R}\times\mathbb{R}\to G$, whose dynamics in $t$ and $s$ may be obtained from Hamilton's principle for a $G$-invariant reduced Lagrangian $l: \mathfrak{g}\times\mathfrak{g}\to\mathbb{R}$, where $\mathfrak{g}$ is the Lie algebra of the group $G$. 
The $G$-strand system of evolutionary partial differential equations for a left $G$-invariant reduced Lagrangian consists of the the  \emph{zero-curvature} equation (\ref{zero-curv}) and the Euler-Poincar\'e (EP) variational equations (\ref{EPst-eqn}). 

\end{definition}

\begin{remark}
Subclasses of the $G$-strand maps contain the harmonic maps into Lie groups studied mathematically in \cite{Uh1989} and the principal chiral models of field theory in theoretical physics, reviewed, e.g., in \cite{Man1994}. An interpretation of the $G$-strand equations as  the dynamics of a continuous spin chain is given in \cite{Ho2011}. This is the source of the term, `strand'. See also \cite{ElGBHoPuRa2010}.
\end{remark}

The evolutionary $G$-strand maps that we consider here arise from  Hamilton's principle $\delta S=0$ with $S=\int \ell\,dt$ for a left $G$-invariant Lagrangian $\ell: \mathfrak{g} \times \mathfrak{g} \to \mathbb{R}$ given by
\begin{eqnarray}
S=\int_a^b \!\!\!\int_{-\infty}^\infty\!\! \ell({\mathsf{X}},{\mathsf{Y}})\,d{s}\,dt
\,,
\end{eqnarray}
with $(\mathsf{X}, \mathsf{Y})\in \mathfrak{g}\times\mathfrak{g}$
where  $\mathfrak{g}$ is the left Lie algebra  of the group $G$.
That is, $\mathsf{X}$ and $\mathsf{Y}$ are left-trivialized tangent vectors at the identity of the group, expressible as
\begin{eqnarray}
{\mathsf{X}}(t,{s})=g^{-1}\partial_t g(t,{s})
\quad\hbox{and}\quad
{\mathsf{Y}}(t,{s})=g^{-1}\partial_{s} g(t,{s})
\,.
\label{DefXnY}
\end{eqnarray}

\begin{remark}
The distinction between between the maps 
$(\mathsf{X},\mathsf{Y}): \mathbb{R}\times \mathbb{R}\to \mathfrak{g}\times\mathfrak{g}$ and their pointwise values $(\mathsf{X}(t,s),\mathsf{Y}(t,s))\in  \mathfrak{g}\times\mathfrak{g}$ will always be clear in context, so that no confusion will arise.
Likewise, for the variational derivatives $\frac{\delta\ell}{\delta {\mathsf{X}}}$ and $\frac{\delta\ell}{\delta {\mathsf{Y}}}$. 
\end{remark}\vspace{-5mm}

\paragraph{Overview and organization of the paper.} After this introduction, we shall begin developing $G$-strand dynamics in the example of the rotation group $SO(3)$ in Section \ref{SO3-sec}. In a certain case, this example recovers the well-known completely integrable principal chiral model \cite{ZaMi1978,ZaMi1980}. We extend these considerations in Section \ref{SO3K-sec} to the isotropy group $SO(3)_K$ of a quadratic form defined by a symmetric matrix, ${\sf K}={\sf K}^T$. In Section \ref{Pchiral-sec}, the $SO(3)_K$-strand equations for a certain choice of Lagrangian are identified with the integrable dynamics associated with another variant of the principal chiral model, the completely integrable $P$-chiral model of \cite{BoYa1995,Ya1988}. Section \ref{SP2-sec} considers $G$-strand dynamics for $Sp(2)$, the two-dimensional symplectic group, and extends the Bloch-Iserles ordinary differential equations to the case of partial differential equations for the $Sp(2)$-strand. By using the `gauge' transformation that diagonalizes  ${\sf K}$ and a series of canonical isomorphisms, the $Sp(2)$-strand equations for a particular choice of Lagrangian are also brought into the form of the integrable $SO(3)_K$ P-chiral model. In each of these cases, we apply the Euler-Poincar\'e (EP) framework for group invariant Lagrangians, then pass to the corresponding Lie-Poisson Hamiltonian framework, where we are able to identify Hamiltonians that produce integrable dynamics. Numerical solutions for the $O(3)$ and $Sp(2)$ versions of the $P$-Chiral models are provided in Section \ref{JRPnumerics}. In Section \ref{DiffStrand-sec} we extend the $G$-strand framework to the case of diffeomorphisms on the real line and show that the Diff$(\mathbb{R})$-strand dynamics admits singular solutions associated with a pair of momentum maps. Finally, Section \ref{conclusion-sec} summarizes our results and provides some outlook for future research.

In the standard Euler-Poincar\'e (EP) framework, one applies the following steps \cite{Ho2011,MaRa1999}.
\begin{description}
\item [(a)]
Write the auxiliary equation for the evolution of ${\mathsf{Y}}:\,\mathbb{R}\times\mathbb{R}\to \mathfrak{g}$, obtained by differentiating its definition with respect to time and invoking equality of cross derivatives.

\item [(b)]
Use the Euler-Poincar\'e theorem for left-invariant Lagrangians to obtain the equation of motion for the quantity $\partial\ell/\partial{\mathsf{X}}:\,\mathbb{R}\times\mathbb{R}\to \mathfrak{g}^*$, where $\mathfrak{g}^*$ is the dual Lie algebra. 

The resulting Euler-Poincar\'e equation will be an evolutionary partial differential equation (PDE). We assume homogeneous  boundary conditions on ${\mathsf{X}}(t,{s})$, ${\mathsf{Y}}(t,{s})$ and vanishing endpoint conditions on the variation ${\mathsf{Z}}=g^{-1}\delta g(t,{s})\in\mathfrak{g}$.

\item [(c)]
Legendre transform the Lagrangian to obtain the corresponding Hamiltonian. Differentiate the Hamiltonian and determine its partial derivatives. Write the Euler-Poincar\'e equation in Lie-Poisson Hamiltonian form, in terms of the new ``angular momentum'' variable ${\mathsf{P}}=\delta\ell/\delta{\mathsf{X}}\in\mathfrak{g}^*$.

\item [(d)]
Write the Lie-Poisson Hamiltonian formulation for $G=SO(3)$ in terms of $\mathbb{R}^3$ vectors.

\item [(e)]
Apply the EP procedure to other interesting choices of the Lie group $G$, e.g., $Sp(2)$, and other choices of the Lagrangian $\ell({\mathsf{X}},{\mathsf{Y}})$ in Hamilton's principle.

\end{description}
The EP framework with steps (a)-(e) provides the organization for each of  the sections that follow. 

\section{Euler-Poincar\'e (EP) procedure for $SO(3)$ chiral model}
\label{SO3-sec}

This Euler-Poincar\'e procedure in steps {\bf (a)-(e)} above produces a series of results that are to be described in this section for the $SO(3)$ chiral model. These are the following.

\subsection*{(a) The auxiliary equation}
According to their definitions  in (\ref{DefXnY}),
$
{\mathsf{X}}(t,{s})=g^{-1}\partial_t g(t,{s})
$
and
$
{\mathsf{Y}}(t,{s})=g^{-1}\partial_{s} g(t,{s})
$ are Lie-algebra-valued functions over $\mathbb{R}\times\mathbb{R}$.
The evolution of ${\mathsf{Y}}$ is obtained from these definitions by taking the difference of the two equations for the partial derivatives
$\partial_t{\mathsf{Y}}(t,{s})$ and $\partial_{s}{\mathsf{X}}(t,{s})$
while invoking equality of cross derivatives. Hence, ${\mathsf{Y}}$ evolves by the adjoint operation,
\begin{equation}
\partial_t{\mathsf{Y}}(t,{s}) - \partial_{s}{\mathsf{X}}(t,{s})
= {\mathsf{Y}} \,{\mathsf{X}} - {\mathsf{X}}\,{\mathsf{Y}}
= [{\mathsf{Y}},\, {\mathsf{X}}]
=: -  {\rm ad}_{\mathsf{X}}{\mathsf{Y}}
\,.
\label{aux-eqn-2time}
\end{equation}
This is the auxiliary equation we seek for ${\mathsf{Y}}(t,{s})$.
In differential geometry, this relation is called a {\bfi zero-curvature equation}, because it implies that the curvature vanishes for the Lie-algebra-valued connection one-form 
\[
A=g^{-1}dg(t,{s})={\mathsf{X}} dt + {\mathsf{Y}} d{s}
\,.
\]
Equation (\ref{aux-eqn-2time}) is also expressible in terms of the connection 1-form $A$ as the Mauer-Cartan relation,
\begin{equation}
dA + A\wedge A = 0 
\,.
\label{chiral-eqn1}
\end{equation}
See, e.g., \cite{doCarmo1976} for further discussion of such matters in differential geometry. 

When augmented by the divergence relation $d*A=0$, namely, 
\begin{equation}
d*A = d({\mathsf{X}} d{s} + {\mathsf{Y}} dt) 
= (\partial_t \mathsf{X} - \partial_{s} \mathsf{Y})\, dt\wedge d{s}
= 0
\,,
\label{chiral-eqn2}
\end{equation}
equations (\ref{chiral-eqn1}) and (\ref{chiral-eqn2}) comprise the {\bfi classical chiral model}.  Equation (\ref{chiral-eqn2}) may also be understood as a harmonic map from $(t,{s})\in \mathbb{R}\times\mathbb{R}$ into the Lie group $G$ \cite{Uh1989}.

\subsection*{(b) The Euler-Poincar\'e theorem}
According to the Euler-Poincar\'e theorem, Hamilton's principle $\delta S=0$ for
$
S=\int_{t_1}^{t_2}\!\!\int_{s_1}^{s_2} \ell({\mathsf{X}},{\mathsf{Y}})\,ds\,dt
$
leads to the result in equation (\ref{EPst-calc}),
\begin{equation}
\frac{\partial}{\partial t} \frac{\delta\ell}{\delta {\mathsf{X}}}
- {\rm ad}^*_{\mathsf{X}}\frac{\delta\ell}{\delta {\mathsf{X}}}
=
- \frac{\partial}{\partial {s}}  \frac{\delta\ell}{\delta {\mathsf{Y}}}
+ {\rm ad}^*_{\mathsf{Y}}\frac{\delta\ell}{\delta {\mathsf{Y}}}
\,.
\label{2timeEP}
\end{equation}
This is the Euler-Poincar\'e equation for $\delta\ell/\delta{\mathsf{X}}\in\mathfrak{g}^*$ and it involves $\delta\ell/\delta{\mathsf{Y}}\in\mathfrak{g}^*$.

\begin{proposition}
The Euler-Poincar\'e equation (\ref{2timeEP}) may be written as a {\bfi conservation law} for angular momentum ${\mathsf{P}}=\delta\ell/\delta {\mathsf{X}}$; viz,
\begin{equation}
\frac{\partial}{\partial t}
 \bigg({\Ad}^*_{g(t,{s})^{-1}} \frac{\delta l}{\delta {\mathsf{X}}}\bigg)
=
-\,\frac{\partial}{\partial {s}}
 \bigg({\Ad}^*_{g(t,{s})^{-1}} \frac{\delta l}{\delta {\mathsf{Y}}}\bigg)
 \,.
\label{cons-spinangmom}
\end{equation}

\end{proposition}
\begin{proof}
This formula follows from a standard result \cite{Ho2011}.
Namely, let $g(s)$ be a path with parameter $s$ in a Lie group $G$
and let ${\mathsf{Q}}(s)$ be a path in the dual $\mathfrak{g}^*$ of its Lie algebra $\mathfrak{g}$. Then the coadjoint operation  ${\rm Ad}^*: G \times \mathfrak{g}^* \to \mathfrak{g}^*$ satisfies
\begin{equation}
\frac{d}{ds} {\Ad}^*_{g(s)^{-1}} {\mathsf{Q}}(s)=
{\Ad}^*_{g(s)^{-1}}\left[ \frac{d {\mathsf{Q}}}{ds} - {\ad}^*_{{\mathsf{X}}(s)}{\mathsf{Q}}(s)\right]
,
\label{coAd-dot}
\end{equation}
where ${\mathsf{X}}(s)=g(s)^{-1}\dot{g}(s)\in\mathfrak{g}$. Hence, one finds the conservation law (\ref{cons-spinangmom}) by using (\ref{coAd-dot}) on either side of equation (\ref{2timeEP}).
\end{proof}

\paragraph{Summary.} In summary, the $G$-strand equations form a system of partial differential equations (PDE) comprising the Euler-Poincar\'e equation (\ref{2timeEP}) and its auxiliary zero-curvature equation (\ref{aux-eqn-2time}), as follows, 
\begin{eqnarray}
\frac{\partial}{\partial t} \frac{\delta\ell}{\delta {\mathsf{X}}}
+ \frac{\partial}{\partial {s}}  \frac{\delta\ell}{\delta {\mathsf{Y}}}
&=& {\rm ad}^*_{\mathsf{X}}\frac{\delta\ell}{\delta {\mathsf{X}}}
+ {\rm ad}^*_{\mathsf{Y}}\frac{\delta\ell}{\delta {\mathsf{Y}}}
\,,
\label{Gstrand-eqn1}
\\ \nonumber\\
\frac{\partial}{\partial t} {\mathsf{Y}} - \frac{\partial}{\partial {s}} {\mathsf{X}}
&=& - \, {\rm ad}_{\mathsf{X}}{\mathsf{Y}}
\,.
\label{Gstrand-eqn2}
\end{eqnarray}

\begin{remark}
The generalizations of the $G$-strand equations (\ref{Gstrand-eqn1}) and (\ref{Gstrand-eqn2}) to higher spatial dimensions are straightforward. 
\end{remark}

\begin{proposition}
When the Lagrangian $\ell({\mathsf{X}},{\mathsf{Y}}):\mathfrak{g}\to \mathbb{R}$ is given by
\begin{equation}\label{KE-PE-Lag}
\ell({\mathsf{X}},{\mathsf{Y}}) 
=
\frac12|\mathsf{X}|^2 - \frac12|\mathsf{Y}|^2
\,,\end{equation}
then Hamilton's principle recovers the chiral model equation (\ref{chiral-eqn2}).
\end{proposition}

\begin{proof}
The proof follows by comparing equations (\ref{2timeEP}) and (\ref{chiral-eqn2}) for this choice of Lagrangian.
\end{proof}

\subsection*{(c) Lie-Poisson Hamiltonian formulation}

Legendre transforming the Lagrangian
$\ell({{\mathsf{X}},{\mathsf{Y}}}):\,
\mathfrak{g}\times \mathfrak{g}\to\mathbb{R}$ yields the following Hamiltonian
$h({{\mathsf{P}},{\mathsf{Y}}}):\,
\mathfrak{g}^*\times \mathfrak{g}\to\mathbb{R}$
\begin{equation}
h({{\mathsf{P}}},{{\mathsf{Y}}})
=
\Big\langle{{\mathsf{P}}}\,,\,{{\mathsf{X}}}\Big\rangle
-
\ell({{\mathsf{X}},{\mathsf{Y}}})
\,.
\label{leglagham}
\end{equation}
Its partial derivatives are determined from
\begin{eqnarray*}
\delta{h}
&=& \Big\langle\,\delta{\mathsf{P}}\,,\, \frac{\delta h}{\delta{\mathsf{P}}}\,\Big\rangle
+
\Big\langle\,\frac{\delta h}{\delta{\mathsf{Y}}}\,,\, \delta{\mathsf{Y}}\,\Big\rangle
\\
&=& \Big\langle\,\delta{\mathsf{P}}\,,\, {\mathsf{X}}\,\Big\rangle
+
\Big\langle\,{\mathsf{P}} - \frac{\delta l}{\delta {\mathsf{X}}}
\,,\,\delta {\mathsf{X}}\,\Big\rangle
-\Big\langle\,\frac{\delta \ell}{\delta{\mathsf{Y}}}\,,\, \delta{\mathsf{Y}}\,\Big\rangle
\,,\\
&\Rightarrow&
\frac{\delta l}{\delta {\mathsf{X}}} = {\mathsf{P}}
\,,\quad
\frac{\delta h}{\delta {\mathsf{P}}} = {\mathsf{X}}
\quad\hbox{and}\quad
\frac{\delta h}{\delta {\mathsf{Y}}} 
= -\,\frac{\delta \ell}{\delta {\mathsf{Y}}} = \mathsf{Y}
.
\end{eqnarray*}
Vanishing of the middle term in the second line defines the momentum ${\mathsf{P}}\in\mathfrak{g}^*$.
These derivatives allow one to rewrite the Euler-Poincar\'e equation solely in terms of momentum ${\mathsf{P}}$ as
\begin{eqnarray}
{\partial_t} {\mathsf{P}}
&=& {\rm ad}^*_{\delta h/\delta {\mathsf{P}}}\, {\mathsf{P}}
+ \partial_{s} \frac{\delta h}{\delta {\mathsf{Y}}}
- {\rm ad}^*_{\mathsf{Y}}\,\frac{\delta h}{\delta {\mathsf{Y}}}
\,,\nonumber\\
\partial_t {\mathsf{Y}}
&=& \partial_{s}\frac{\delta h}{\delta {\mathsf{P}}}
-  {\rm ad}_{\delta h/\delta {\mathsf{P}}}\,{\mathsf{Y}}
\,.
\label{hameqns-so3}
\end{eqnarray}
%


The corresponding evolution equation for any functional of $f({\mathsf{P}},{\mathsf{Y}})$ then follows as
\begin{eqnarray*}
\frac{\partial}{\partial t}f({\mathsf{P}},{\mathsf{Y}})
&=&
\Big\langle\,{\partial_t}{\mathsf{P}}\,,\,\frac{\delta f}{\delta {\mathsf{P}}}\,\Big\rangle
+
\Big\langle\,{\partial_t}{\mathsf{Y}}\,,\,\frac{\delta f}{\delta {\mathsf{Y}}}\,\Big\rangle
\\
&=&
 \Big\langle\,{\rm ad}^*_{\delta h/\delta {\mathsf{P}}} {\mathsf{P}}
+
 \partial_{s} \frac{\delta h}{\delta {\mathsf{Y}}}
- {\rm ad}^*_{\mathsf{Y}}\frac{\delta h}{\delta {\mathsf{Y}}}\,,\,
\frac{\delta f}{\delta{\mathsf{P}}}\,\Big\rangle
\\&&+\
\Big\langle\,\partial_{s}\frac{\delta h}{\delta {\mathsf{P}}}
-  {\rm ad}_{\delta h/\delta {\mathsf{P}}}{\mathsf{Y}}\,,\,\frac{\delta f}{\delta {\mathsf{Y}}}\,\Big\rangle
\\
&=&
-\,\Big\langle\,{\mathsf{P}}\,,\,\bigg[
\frac{\delta f}{\delta {\mathsf{P}}}\,,\,\frac{\delta h}{\delta {\mathsf{P}}}\bigg]\,\Big\rangle
\\
&&+\
\Big\langle\, \partial_{s} \frac{\delta h}{\delta {\mathsf{Y}}}
\,,\,
\frac{\delta f}{\delta{\mathsf{P}}}\,\Big\rangle
-
\Big\langle\, \partial_{s} \frac{\delta f}{\delta {\mathsf{Y}}}
\,,\,
\frac{\delta h}{\delta{\mathsf{P}}}\,\Big\rangle
\\&&+\
\Big\langle\,{\mathsf{Y}}
\,,\,
{\rm ad}^*_{\delta f/\delta {\mathsf{P}}}\,\frac{\delta h}{\delta {\mathsf{Y}}}
-  {\rm ad}^*_{\delta h/\delta {\mathsf{P}}}\,\frac{\delta f}{\delta {\mathsf{Y}}}
\,\Big\rangle
\\
&=:&
\{f\,,\,h\}({\mathsf{P}},{\mathsf{Y}})
\,,
\end{eqnarray*}
in which the bracket $\{\,\cdot\,,\,\cdot\,\}$ satisfies the properties defining a Poisson bracket because it is dual to the Lie bracket.
Assembling these equations into
Hamiltonian form gives,
%
\begin{equation} \label{LP-Ham-struct-symbol}
\frac{\partial}{\partial t}
    \begin{bmatrix}
    {\mathsf{P}}
    \\
    {\mathsf{Y}}
    \end{bmatrix}
=
\begin{bmatrix}
  {\rm ad}^\ast_\square {\mathsf{P}}
   &
  \partial_s - {\rm ad}^*_{\mathsf{Y}}
   \\
   \partial_s + {\rm ad}_{\mathsf{Y}}
   & 0
    \end{bmatrix}
    \begin{bmatrix}
   \delta h/\delta{\mathsf{P}} \\
   \delta h/\delta{\mathsf{Y}}
    \end{bmatrix}
\end{equation}
in which the box $\square$ indicates how the ad- and ad$^*$\!- operations occur in the matrix multiplication. In particular,
\[
{\rm ad}^\ast_\square {\mathsf{P}}\,(\delta h/\delta{\mathsf{P}})= {\rm ad}^\ast_{\delta h/\delta{\mathsf{P}}} {\mathsf{P}}
\,,
\]
so each matrix entry acts on its corresponding vector component.%
\footnote{The matrix in equation (\ref{LP-Ham-struct-symbol}) is also the lower right corner of the Hamiltonian matrix for a perfect complex fluid \cite{Ho2002}. It also appears in the Lie-Poisson brackets for Yang-Mills fluids \cite{GiHoKu1982,GiHoKu1983} and for spin glasses \cite{HoKu1988}. See those references for full discussions of the Lie algebraic meaning of such Lie-Poisson brackets for the perfect complex fluid.}

\subsubsection*{Higher dimensions (G-branes)}
\label{secHigherDimensions}
The notation in the Hamiltonian matrix (\ref{LP-Ham-struct-symbol}) indicates how the jump to higher dimensions in the variable $s\in\mathbb{R}^n$ may be made naturally. This is done by using the gradient $\partial_j=\partial/\partial s^j$, $j=1,2,\dots,n$ to define the left invariant auxiliary variable ${\mathsf{Y}}_j\equiv g^{-1}\partial_j g$ in higher dimensions. The lower left entry of the matrix (\ref{LP-Ham-struct-symbol}) defines a covariant  gradient, and its upper right entry defines the adjoint operator, the covariant divergence.  More explicitly, in terms of indices and partial differential
operators, this Hamiltonian matrix becomes,
\begin{equation}
\frac{\partial}{\partial t}
    \begin{bmatrix}
    {\mathsf{P}}_\alpha
    \\
    {\mathsf{Y}}_i^\alpha
    \end{bmatrix}
=
    B_{\alpha\beta}
    \begin{bmatrix}
   \delta h/\delta{\mathsf{P}}_\beta \\
   \delta h/\delta{\mathsf{Y}}_j^\beta
    \end{bmatrix},
    \label{LP-Ham-struct-diff1}
\end{equation}
where the Hamiltonian structure matrix $B_{\alpha\beta}$ is given explicitly as  \cite{Ho2002},
\begin{equation}
B_{\alpha\beta}
=
\begin{bmatrix}
    -{\mathsf{P}}_\kappa \,t_{\alpha\beta}^{\,\kappa}
   &     \delta_\alpha^{\,\beta}\partial_j
   + t_{\alpha\kappa}^{\,\beta} {\mathsf{Y}}_j^{\,\kappa}
      \\
 \delta_\beta^\alpha\partial_i
   -  t_{\beta \kappa}^{\,\alpha}{\mathsf{Y}}_i^ \kappa
& 0
    \end{bmatrix}.
    \label{LP-Ham-struct-diff2}
\end{equation}
Here, the summation convention is enforced on repeated indices. Upper Greek indices refer to the Lie algebraic basis set, lower Greek indices refer to the dual basis and Latin indices refer to the spatial reference frame. The partial derivatives $\partial_i=\partial/\partial s^i$, $i=1,2,\dots,n$, act to the right on all terms in a product by the chain rule.
%
%

\subsection*{(d) Write the Lie-Poisson Hamiltonian form for $G=SO(3)$}

For the case that $t^{\,\alpha}_{\ \beta\kappa}$ are structure constants
for the Lie algebra $\mathfrak{so}(3)$, then $t^{\,\alpha}_{\ \beta\kappa}=\epsilon_{\alpha\beta\kappa}$ with $\epsilon_{123}=+1$ and the Lie bracket on $\mathfrak{so}(3)$ may be identified with the vector product on $\mathbb{R}^3$. The  Lie-Poisson Hamiltonian matrix in (\ref{LP-Ham-struct-diff2}) may then be rewritten in vector form as
\begin{equation} \label{LP-Ham-struct-so3}
\frac{\partial}{\partial t}
    \begin{bmatrix}
    \boldsymbol{{\mathsf{P}}}
    \\
    \boldsymbol{{\mathsf{Y}}}_i
    \end{bmatrix}
=
\begin{bmatrix}
    \boldsymbol{{\mathsf{P}}}\times
   &     \partial_j
   + \boldsymbol{{\mathsf{Y}}}_j\times
      \\
    \partial_i
   + \boldsymbol{{\mathsf{Y}}}_i\times
& 0
    \end{bmatrix}
    \begin{bmatrix}
   \delta h/\delta\boldsymbol{{\mathsf{P}}} \\
   \delta h/\delta\boldsymbol{{\mathsf{Y}}}_j
    \end{bmatrix}.
\end{equation}
Returning to one dimension, one sees that stationary solutions for either $\partial_t\to0$, or $\partial_s\to0$ satisfy equations of the same form as the heavy top. That the equations for either type of solution would have the same form might be expected, because of the exchange symmetry under $t\leftrightarrow s$ and $\boldsymbol{\mathsf{X}}\leftrightarrow \boldsymbol{\mathsf{Y}}$. Perhaps less expected is that these equations would appear in the form of a  heavy top. 
The EP formulation of the  heavy top equations appears, for example, in \cite{Ho2011}.

For $G=SO(3)$ and Lagrangian
$
\ell (\boldsymbol{{\mathsf{X}},\,{\mathsf{Y}}} ): \mathbb{R}^3\times\mathbb{R}^3\to\mathbb{R},
$
 in one spatial dimension the  Euler-Poincar\'e equation and its Hamiltonian form are given in terms of vector operations in $\mathbb{R}^3$, as follows. First, the
Euler-Poincar\'e equation (\ref{2timeEP}) becomes
\begin{equation}
\frac{\partial}{\partial t} \frac{\delta\ell}{\delta \boldsymbol{{\mathsf{X}}}}
= -\,\boldsymbol{{\mathsf{X}}}\times
\frac{\delta\ell}{\delta \boldsymbol{{\mathsf{X}}} }
- \frac{\partial}{\partial {s}}  \frac{\delta\ell}{\delta \boldsymbol{{\mathsf{Y}}} }
-\,{\boldsymbol{{\mathsf{Y}}}}\times
\frac{\delta\ell}{\delta \boldsymbol{{\mathsf{Y}}} }
\,,
\label{2timeEP-SO3}
\end{equation}
and its auxiliary equation (\ref{aux-eqn-2time}) becomes
\begin{equation}
\frac{\partial}{\partial t} \boldsymbol{{\mathsf{Y}}}
=
\frac{\partial}{\partial {s}}{\boldsymbol{{\mathsf{X}}}}
+
{\boldsymbol{{\mathsf{Y}}}}\times{\boldsymbol{{\mathsf{X}}}}
\,.
\label{aux-eqn-2timeX}
\end{equation}


The Hamiltonian structures of these equations on $\mathfrak{so}(3)^*$ are then obtained from the Legendre transform relations
\begin{eqnarray*}
\frac{\delta \ell}{\delta \boldsymbol{{\mathsf{X}}} }
= \boldsymbol{{\mathsf{P}}}
\,,\quad
\frac{\delta h}{\delta \boldsymbol{{\mathsf{P}}}} = \boldsymbol{{\mathsf{X}}}
\quad\hbox{and}\quad
\frac{\delta h}{\delta \boldsymbol{{\mathsf{Y}}} }
= -\,\frac{\delta \ell}{\delta \boldsymbol{{\mathsf{Y}}} }
\,.
\end{eqnarray*}
Hence, the  Euler-Poincar\'e equation (\ref{2timeEP}) becomes
\begin{equation}
\frac{\partial}{\partial t} \boldsymbol{{\mathsf{P}}}
=
\boldsymbol{{\mathsf{P}}} \times \frac{\delta h}{\delta \boldsymbol{{\mathsf{P}}}}
+ \frac{\partial}{\partial {s}}  \frac{\delta h}{\delta \boldsymbol{{\mathsf{Y}}} }
+\,{\boldsymbol{{\mathsf{Y}}}}\times
\frac{\delta h}{\delta \boldsymbol{{\mathsf{Y}}} }
\,,
\label{2timeHam-SO3}
\end{equation}
and the auxiliary equation (\ref{aux-eqn-2time}) becomes
\begin{equation}
\frac{\partial}{\partial t} \boldsymbol{{\mathsf{Y}}}
=
\frac{\partial}{\partial {s}}\frac{\delta h}{\delta \boldsymbol{{\mathsf{P}}}}
+
{\boldsymbol{{\mathsf{Y}}}}\times\frac{\delta h}{\delta \boldsymbol{{\mathsf{P}}}}
\,,
\label{aux-eqn-2timeX}
\end{equation}
which recovers the Lie-Poisson Hamiltonian form in equation (\ref{LP-Ham-struct-so3}).

Finally, the equations for reconstructing $O(t,{s})\in SO(3)$ from the evolution of $\boldsymbol{{\mathsf{X}}}(t,{s})$ and $\boldsymbol{{\mathsf{Y}}}(t,{s})$ may be expressed by using the hat map $\,$ $\widehat{\phantom{}}\,: \mathbb{R}^3\to \mathfrak{so}(3)$, e.g., $\widehat{{\mathsf{X}}}_{ij}=-\,\epsilon_{ijk}{X}^k$, as
\begin{eqnarray}
&&\partial_t O(t,{s})=O(t,{s})\widehat{{\mathsf{X}}}(t,{s})
\quad\hbox{and}
\nonumber\\&&
\partial_{s} O(t,{s})=O(t,{s})\widehat{{\mathsf{Y}}}(t,{s})
\,.
\label{reconstruct-eqn-2time}
\end{eqnarray}

\begin{remark}
The Euler-Poincar\'e equations for the $G$-strand discussed here and their Lie-Poisson Hamiltonian formulation provide a framework for systematically investigating three-dimensional orientation dynamics along a fixed one-dimensional string. These partial differential equations are interesting in their own right and they have many possible applications. For an idea of where the applications of these equations could lead if the string were also allowed to bend self-consistently, consult  \cite{SiMaKr1988, ElGBHoPuRa2010}.
\end{remark}

The remainder of the paper will address the final item (e) in the Euler-Poincar\'e $G$-strand procedure, by making two other choices for the group and the Lagrangian. 

%


\section{The $SO(3)_K$ strand}\label{SO3K-sec}

Another interesting choice for the Lagrangian extends the orthogonal group $SO(3)$ to the isotropy group of a symmetric quadratic form in $n$ dimensions.

\begin{proposition}[Isotropy group]\label{IsotropSubgrp-prop}
Let $K \in GL(3,\mathbb{R})$ be a symmetric matrix, $K=K^T$.
The set
\begin{equation}
SO(3)_K = \{U \in GL(3,\mathbb{R})|\,U^TKU = K\}
\,,
\end{equation}
defines the subgroup $SO(3)_K$ of $GL(3,\mathbb{R})$.
Moreover, this isotropy group of $K$ is a submanifold of $\mathbb{R}^{n\times n}$ of dimension  $n(n-1)/2$.
\end{proposition}

\begin{proposition}\label{Lie-algebra-prop}
Let  $\widehat{{\mathsf{X}}}=U^{-1}\dot{U}$ and $\widehat{{\mathsf{Y}}}=U^{-1}{U'}$ for $U\in S$, where ${U'}$ and $\dot{U}$ denote derivatives of $U$ in $s$ and $t$, respectively. Then the matrix quantities
${\mathsf{X}}:=K\widehat{{\mathsf{X}}}$ and ${\mathsf{Y}}:=K\widehat{{\mathsf{Y}}}$ are antisymmetric
\begin{equation}
{\mathsf{X}}^T = -\,{\mathsf{X}}
\quad\hbox{and}\quad
{\mathsf{Y}}^T = -\,{\mathsf{Y}}
\,.
\end{equation}

\end{proposition}

\begin{proof}
This is a direct verification.
\end{proof}
Let $\widehat{{\mathsf{X}}}'$ and $\widehat{{\mathsf{Y}}}\dot{\phantom{\,\,}}$ denote separate derivatives of $\widehat{{\mathsf{X}}}$ and $\widehat{{\mathsf{Y}}}$.
Then, by an easy calculation one finds,
\[
\widehat{{\mathsf{X}}}'  - \widehat{{\mathsf{Y}}}\dot{\phantom{\,\,}}
=
\widehat{{\mathsf{X}}}\,\widehat{{\mathsf{Y}}}
-
\widehat{{\mathsf{Y}}}\,\widehat{{\mathsf{X}}}
=:  [\widehat{{\mathsf{X}}},\widehat{{\mathsf{Y}}}]
.\]
Hence, upon substituting the definitions of ${\sf X}$ and ${\sf Y}$, one finds
\begin{eqnarray*}
{\sf X}' = K\widehat{{\mathsf{X}}}'
&=&
K \widehat{{\mathsf{Y}}}\dot{\phantom{\,\,}}
+
K[\widehat{{\mathsf{X}}},\widehat{{\mathsf{Y}}}]
\\
&=&
\dot{{\sf Y}} + {\sf X}K^{-1}{\sf Y} - {\sf Y}K^{-1}{\sf X}
\\
&=:&
\dot{{\sf Y}} + [{\sf X},{\sf Y}]_K
\\
&=:&
\dot{{\sf Y}} + \ad^K_{\sf X}\,{\sf Y}
.\end{eqnarray*}
The corresponding auxiliary equation is to be compared with (\ref{aux-eqn-2time}),
\begin{equation}
\partial_t{\mathsf{Y}}(t,{s}) - \partial_{s}{\mathsf{X}}(t,{s})
=
-  {\rm ad}^K_{\mathsf{X}}{\mathsf{Y}}
\,,
\label{aux-eqn-2timeK}
\end{equation}
in which ${\mathsf{Y}}$ evolves by the $K$-adjoint operation.

The same calculation of Hamilton's principle $\delta S=0$ as in the first section now yields the following Euler-Poincar\'e equation, which is to be compared with (\ref{2timeEP}),
\begin{equation}
\frac{\partial}{\partial t} \frac{\delta\ell}{\delta {\mathsf{X}}}
= {\rm ad}^{K*}_{\mathsf{X}}\frac{\delta\ell}{\delta {\mathsf{X}}}
- \frac{\partial}{\partial {s}}  \frac{\delta\ell}{\delta {\mathsf{Y}}}
+ {\rm ad}^{K*}_{\mathsf{Y}}\frac{\delta\ell}{\delta {\mathsf{Y}}}
\,,
\label{2timeEPK}
\end{equation}
in which $\delta\ell/\delta {\mathsf{X}}$ evolves by the $K$-coadjoint operation, defined as
\begin{equation}
\Big\langle
{\rm ad}^{K*}_{\mathsf{X}}\frac{\delta\ell}{\delta {\mathsf{X}}}
\,,\,{\mathsf{Z}} \Big\rangle
=
\Big\langle \frac{\delta\ell}{\delta {\mathsf{X}}}
\,,\,{\rm ad}^K_{\mathsf{X}}{\mathsf{Z}} \Big\rangle
\,.
\label{Kad-star-def}
\end{equation}

Upon specializing to $n=3$ and identifying the Lie bracket on $\mathfrak{so}(3)$ with the vector product on $\mathbb{R}^3$ as we did before, we find the corresponding $K$-vector equations for the motion on the isotropy group $S$, by mapping
\begin{equation}
 {\rm ad}^K_{\mathsf{X}}{\mathsf{Y}} \to K({\boldsymbol{{\mathsf{X}}}}\times{\boldsymbol{{\mathsf{Y}}}})
 \quad\hbox{and}\quad
 {\rm ad}^{K*}_{\mathsf{X}}{\mathsf{P}} \to - \,{\boldsymbol{{\mathsf{X}}}}\times{K\boldsymbol{{\mathsf{P}}}}
\,.
\label{EPKross}
\end{equation}
This means the  Lie-Poisson Hamiltonian matrix in (\ref{LP-Ham-struct-diff2}) may be rewritten in $K$-vector form as
\begin{equation} \label{LP-Ham-struct-Kso3}
\frac{\partial}{\partial t}
    \begin{bmatrix}
    \boldsymbol{{\mathsf{P}}}
    \\
    \boldsymbol{{\mathsf{Y}}}
    \end{bmatrix}
=
\begin{bmatrix}
    K\boldsymbol{{\mathsf{P}}}\times\Box
   &     \partial_{s}
   + \boldsymbol{{\mathsf{Y}}}\times (K\Box)
      \\
    \partial_{s}
   + K(\boldsymbol{{\mathsf{Y}}}\times\Box)
& 0
    \end{bmatrix}
    \begin{bmatrix}
   \delta h/\delta\boldsymbol{{\mathsf{P}}} \\
   \delta h/\delta\boldsymbol{{\mathsf{Y}}}
    \end{bmatrix}.
\end{equation}

In the next section, the $K$-vector equations obtained this way will be compared with the completely integrable $P$-chiral model \cite{BoYa1995,Ya1988}.

\section{Relation to an integrable system: the $P$-chiral model}\label{Pchiral-sec}

For the Lagrangian
$\ell=\frac{1}{2}(|\boldsymbol{\mathsf{X}}|^2-|\boldsymbol{\mathsf{Y}|}^2)$ in (\ref{KE-PE-Lag}) the equations
(\ref{aux-eqn-2timeK}) - (\ref{EPKross}) in vector form are

\begin{equation}
\partial_t{\boldsymbol{\mathsf{Y}}}(t,{s}) - \partial_{s}{\boldsymbol{\mathsf{X}}}(t,{s})
+K ({\boldsymbol{\mathsf{X}}}\times{\boldsymbol{\mathsf{Y}}})=0 \,, \label{XY-eqn-1}
\end{equation}

\begin{equation}
\partial_{s}{\boldsymbol{\mathsf{Y}}}(t,{s}) - \partial_t{\boldsymbol{\mathsf{X}}}(t,{s})
+ {\boldsymbol{\mathsf{Y}}}\times K \boldsymbol{{\mathsf{Y}}}-{\boldsymbol{\mathsf{X}}}\times K
{\boldsymbol{\mathsf{X}}}=0 \,.
\label{XY-eqn-2}
\end{equation}
These may be cast into Lie-Poisson Hamiltonian form, as follows. 
%
\subsection{Hamiltonian formulation of the $P$-chiral model in
$\mathbb{R}^3\times\mathbb{R}^3$}

The Legendre transformation (\ref{leglagham}) in this case yields $\boldsymbol{{\mathsf{P}}}=\boldsymbol{{\mathsf{X}}}$ and
\begin{equation}
h(\boldsymbol{{\mathsf{P}}},\boldsymbol{{\mathsf{Y}}})
=
\frac12 |\boldsymbol{{\mathsf{P}}}|^2
+ \frac12 |\boldsymbol{{\mathsf{Y}}}|^2
\,,
\label{chiral-Ham}
\end{equation}
which is sign-definite.  Equation (\ref{LP-Ham-struct-Kso3}) now verifies  the equations of motion in Lie-Poisson Hamiltonian form,
\begin{eqnarray} \label{LP-Ham-struct-Kso3-H}
{\partial_t} \boldsymbol{{\mathsf{P}}}
&=&
    K\boldsymbol{{\mathsf{P}}}\times\boldsymbol{{\mathsf{P}}}
   +     \partial_{s}\boldsymbol{{\mathsf{Y}}}
   + \boldsymbol{{\mathsf{Y}}}\times (K\boldsymbol{{\mathsf{Y}}})
    \,,\\
    \label{LP-Ham-struct-Kso3-Ha}
{\partial_t}
    \boldsymbol{{\mathsf{Y}}}
&=&
    \partial_{s}\boldsymbol{{\mathsf{P}}}
   + K(\boldsymbol{{\mathsf{Y}}}\times\boldsymbol{{\mathsf{P}}})
   \,.
\end{eqnarray}

We introduce a new diagonal matrix $P$, according to
\begin{equation}
P= \frac{1}{2}\text{tr}(K)\mathbf{1} - K \,, \qquad
K=\text{tr}(P)\mathbf{1} - P
. \label{def-P}
\end{equation}
With this definition one may easily check the identity
\begin{equation}
K ({\boldsymbol{\mathsf{X}}}\times{\boldsymbol{\mathsf{Y}}})={\boldsymbol{\mathsf{X}}}\times
P{\boldsymbol{\mathsf{Y}}}-{\boldsymbol{\mathsf{Y}}}\times P {\boldsymbol{\mathsf{X}}}
\,. \label{K-P Lemma}
\end{equation}
By using this identity,  equations (\ref{XY-eqn-1}) - 
(\ref{XY-eqn-2}) can be written as

\begin{equation}
\partial_t{\boldsymbol{\mathsf{Y}}}(t,{s}) - \partial_{s}{\boldsymbol{\mathsf{X}}}(t,{s})
+{\boldsymbol{\mathsf{X}}}\times P{\boldsymbol{\mathsf{Y}}}-{\boldsymbol{\mathsf{Y}}}\times P{\boldsymbol{\mathsf{X}}}=0 \,, \label{XY-eqn-1a}
\end{equation}
\begin{equation}
\partial_{s}{\boldsymbol{\mathsf{Y}}}(t,{s}) - \partial_t{\boldsymbol{\mathsf{X}}}(t,{s})
- {\boldsymbol{\mathsf{Y}}}\times P {\boldsymbol{\mathsf{Y}}}+{\boldsymbol{\mathsf{X}}}\times P{\boldsymbol{\mathsf{X}}}=0 \,. \label{XY-eqn-2a}
\end{equation}

\subsection{Integrability of the $P$-Chiral model}
\label{P-Chiral}
Under the linear change of variables

\begin{equation}
\boldsymbol{\mathsf{u}}
=
\boldsymbol{\mathsf{X}} - \boldsymbol{\mathsf{Y}}
\quad\hbox{and}\quad
\boldsymbol{\mathsf{v}}
=
-\,\boldsymbol{\mathsf{X}} - \boldsymbol{\mathsf{Y}}
\label{changeXY-uv}
\end{equation}
equations (\ref{XY-eqn-1a}) and (\ref{XY-eqn-2a}) acquire the form of the following $O(3)$ $P$-Chiral model, \cite{BoYa1995}, \cite{Ya1988},
\cite{Ge2008}, 

\begin{equation}
\partial_t{\boldsymbol{\mathsf{u}}}(t,{s}) +\partial_{s}{\boldsymbol{\mathsf{u}}}(t,{s})
+{\boldsymbol{\mathsf{u}}}\times P{\boldsymbol{\mathsf{v}}}=0 \,, \label{uv-eqn-1}
\end{equation}

\begin{equation}
\partial_t{\boldsymbol{\mathsf{v}}}(t,{s}) - \partial_{s}{\boldsymbol{\mathsf{v}}}(t,{s})
+{\boldsymbol{\mathsf{v}}}\times P {\boldsymbol{\mathsf{u}}}=0 
\,, \label{uv-eqn-2}
\end{equation}
where the diagonal matrix $P$ is defined in terms of the $3\times3$ symmetric matrix $K$ in equation (\ref{def-P}).

The system (\ref{uv-eqn-1}) - (\ref{uv-eqn-2}) represents two \emph{cross-coupled} equations for
${\boldsymbol{\mathsf{u}}}$ and ${\boldsymbol{\mathsf{v}}}$. 
These equations preserve the magnitudes $|{\boldsymbol{\mathsf{u}}}|^2$ and $|{\boldsymbol{\mathsf{v}}}|^2$, so they allow the further
assumption that the vector fields $(\boldsymbol{\mathsf{u}},\boldsymbol{\mathsf{v}})$ take values on the product of unit spheres
$\mathbb{S}^2 \times \mathbb{S}^2 \subset  \mathbb{R}^3 \times\mathbb{R}^3$. The $P$-Chiral model is an
integrable system and its Lax pair \cite{Lax68} in terms of $(\boldsymbol{\mathsf{u}},\boldsymbol{\mathsf{v}})$ is given in
\cite{BoYa1995}. Expressing its Lax pair  in terms of $(\boldsymbol{\mathsf{X}},\boldsymbol{\mathsf{Y}})$ utilizes the following isomorphism between $\mathfrak{so}(3) \oplus \mathfrak{so}(3)$ and $\mathfrak{so}(4)$:
\begin{equation}
A({\boldsymbol{\mathsf{X}}},{\boldsymbol{\mathsf{Y}}})=\left(
\begin{matrix}
0 & X_3&-X_2 & Y_1\\
-X_3 & 0&X_1 & Y_2\\
X_2 & -X_1&0 & Y_3 \\
-Y_1 & -Y_2&-Y_3 & 0
\end{matrix}
\right) . \label{Mso4-def}
\end{equation}

The system (\ref{XY-eqn-1}) - (\ref{XY-eqn-2}) can be recovered
as a compatibility condition of the operators
\begin{eqnarray}
L&=&\partial_{s}-A({\boldsymbol{\mathsf{Y}}},{\boldsymbol{\mathsf{X}}})(\lambda\,{\rm Id}+J),\\
M&=&\partial_t-A({\boldsymbol{\mathsf{X}}},{\boldsymbol{\mathsf{Y}}})(\lambda\,{\rm Id}+J),
\label{L-Mpair}
\end{eqnarray}
where the diagonal matrix $J$ is defined by
\begin{equation}
J = -\frac{1}{2}\text{diag}(-K_1+K_2+K_3,K_1-K_2+K_3,K_1+K_2-K_3,K_1+K_2+K_3).
\label{J-def}
\end{equation}
These steps get us to the $O(3)$ $P$-chiral model and explain its derivation as an Euler-Poincar\'e equation and thus as a Lie-Poisson system. Previous derivations had identified a linear Poisson structure by other methods, mainly based on Lax pairs and Maurer-Cartan, or zero-curvature relations.  Now that we have taken these steps, we can build a chiral-type model as an Euler-Poincar\'e equation with quadratic kinetic and potential energy on any Lie group. 

\begin{remark}
If ${\sf K}={\rm Id}$, equations (\ref{XY-eqn-1}) - (\ref{XY-eqn-2}) recover the $O(3)$ chiral model.
\end{remark}


\section{The $Sp(2)$ strand}\label{SP2-sec}

\subsection{Bloch-Iserles (BI) equation}
\label{BlochIserles}
The Bloch-Iserles equation is the  $n\times n$ matrix differential system \cite{BlIs2006}
\begin{equation}
\dot{X}(t) = [N, X^2]
,
\label{BI-eqn1}
\end{equation}
where the skew-symmetric matrix $N\in \mathfrak{so}(n)$ is given and the solution $X(t)\in {\rm Sym}(n)$ is a symmetric matrix. 
The BI equation is expressible equivalently as
\begin{equation}
\dot{X}(t) = [B(X), X]
\,,\quad\hbox{with}\quad
B(X) = NX + XN
\,,
\label{BI-eqn2}
\end{equation}
where $B(X): {\rm Sym}(n) \to \mathfrak{so}(n)$. Consequently, the solution of (\ref{BI-eqn2}) for $X(t)$ is given by the similarity transformation
\begin{equation}
{X}(t) = Q(t)X(0)Q^{-1}(t)
\,,\quad\hbox{with}\quad
\dot{Q}Q^{-1}(t) = B(X(t))\in \mathfrak{so}(n)
\,.
\label{BI-eqn3}
\end{equation}
This form of the solution reveals that the BI equation is isospectral, i.e., the eigenvalues of the symmetric matrix $X(t)$ are preserved. Isospectrality leads to constants of motion, obtained from the commutator form of the BI equation with $B(X) = NX + XN$,
\begin{equation}
\dot{X}(t) = [ NX + XN + \lambda N^2, X + \lambda N]
\,,\quad\hbox{with}\quad
\lambda \in \mathbb{R}
\,.
\label{BI-eqn4}
\end{equation}
Consequently, the matrix invariants of $X(t)$ are conserved,
\begin{equation}
{\rm tr}(X + \lambda N)^k = {\rm const}
\,,\quad\hbox{for}\quad
k = 1,2,\dots,n-1
\,.
\label{BI-eqn5}
\end{equation}
For more information about the further analysis of the $n$-component BI equation in general form, see \cite{BlIs2006} and references therein. In what follows, we introduce the simplest case of the BI equation, for the symplectic group $Sp(2)$, then we treat its corresponding $G$-strand extension. 

\subsection{The simplest BI equation}
The simplest nontrivial form of the BI equation may be derived as a particular Euler-Poincar\'e equation on the symplectic group $Sp(2)$.
Euler-Poincar\'e equations on the symplectic group $Sp(2n)$ have been discussed earlier in \cite{GiHoTr2008,GBTr2011} and follow a familiar pattern.

Let the set of $2\times2$ matrices $M_i$ with $i=1,2,3$ satisfy the defining relation for the symplectic Lie group $Sp(2)$,
\begin{equation}
M_iJM_i^T=J
\quad\hbox{with}\quad
J=\left(
\begin{matrix}
0 & -1 \\
1 & 0
\end{matrix}
\right)
\quad\hbox{and no sum on index }i=1,2,3.
\label{sp2group-def}
\end{equation}
The corresponding elements of its Lie algebra $m_i=\dot{M}_iM_i^{-1}\in \mathfrak{sp}(2)$  satisfy $(Jm_i)^T=Jm_i$ for each $i=1,2,3$. Thus, ${\sf X}_i=Jm_i$ satisfying ${\sf X}_i^T={\sf X}_i$ comprises a set of three symmetric $2\times2$ matrices. For definiteness, we may choose a basis given by
\begin{equation}
{\sf X}_1=Jm_1=
\left(
\begin{matrix}
2 & 0 \\
0 & 0
\end{matrix}
\right)
,\qquad
{\sf X}_2=Jm_2=
\left(
\begin{matrix}
0 & 0 \\
0 & 2
\end{matrix}
\right)
,\qquad
{\sf X}_3=Jm_3=
\left(
\begin{matrix}
0 & 1 \\
1 & 0
\end{matrix}
\right).
\label{sp2-basis}
\end{equation}
This basis corresponds to the vector of momentum maps given by quadratic phase-space functions $\mathbf{X}=(|\mathbf{q}|^2,|\mathbf{p}|^2,\mathbf{q}\cdot\mathbf{p})^T$ that are often used in geometric optics and in the theory of particle beam design. That is, for $\mathbf{z}=(\mathbf{q},\mathbf{p})^T$ one identifies
\begin{equation}
\tfrac12\mathbf{z}^T{\sf X}_1\mathbf{z}
=|\mathbf{q}|^2
=X_1
,\qquad
\tfrac12\mathbf{z}^T{\sf X}_2\mathbf{z}
=|\mathbf{p}|^2
=X_2
,\qquad
\tfrac12\mathbf{z}^T{\sf X}_3\mathbf{z}
=\mathbf{q}\cdot\mathbf{p}
=X_3
\,.
\label{sp2-psmomaps}
\end{equation}

\begin{lemma} \label{lemma-JB}
For ${\sf X}=Jm$, ${\sf Y}=Jn\in sym(2)$ with $m,n\in \mathfrak{sp}(2)$,
the J-bracket
\begin{equation}
[{\sf X,Y}]_J
:= {\sf X}J{\sf Y}-{\sf Y}J{\sf X}
=: 2{\rm sym}({\sf X}J{\sf Y})
=: {\ad}^J_{\sf X}{\sf Y}
\label{J-bracket}
\end{equation}
satisfies the Jacobi identity.
\end{lemma}

\begin{proof}
By a straightforward calculation using $J^2=-{\rm Id}_{2\times2}$ with the definitions of ${\sf X}$ and ${\sf Y}$, one finds
\begin{equation}
[{\sf X,Y}]_J:={\sf X}J{\sf Y}-{\sf Y}J{\sf X}=-J(mn-nm)=-J[m,n]
\,.
\label{JB-proof}
\end{equation}
The lemma then  follows from the Jacobi identity for the symplectic Lie algebra and linearity in the definitions of ${\sf X},{\sf Y}\in sym(2)$ in terms of $m,n\in \mathfrak{sp}(2)$.
\end{proof}

\subsection{The BI $Sp(2)$ strand}
\begin{lemma} \label{lemma-xder}
Let ${\sf X}=JM_tM^{-1}$ with time derivative $M_t=\partial_t M(t,{s})$ and let ${\sf Y}=JM' M^{-1}$ with derivative $M'=\partial_s M(t,{s})$. Then
\begin{equation}
 {\sf X}' = \dot{{\sf Y}} + [{\sf X},{\sf Y}]_J
\,,
\label{xder-rel}
\end{equation}
for the J-bracket defined in equation (\ref{J-bracket}).
\end{lemma}

\begin{proof}
Let $m=M_tM^{-1}$ and $n=M'M^{-1}$. Then equality of cross derivatives implies the relation
\[
m'  - \dot{n} = nm-mn =: - [m,n]
\,.
\]
Hence, upon substituting the definitions of the symmetric matrices ${\sf X}$ and ${\sf Y}$, and using $J^T=-J$, one finds
\begin{eqnarray*}
{\sf X}' = Jm'
&=&
J\dot{n} - J[m,n]
\\
&=&
\dot{{\sf Y}} + [{\sf X},{\sf Y}]_J
=
\dot{{\sf Y}} + {\sf X}J{\sf Y}-{\sf Y}J{\sf X}
\\
&=&
\dot{{\sf Y}} + 2{\rm sym}({\sf X}J{\sf Y})
=:
\dot{{\sf Y}} + {\ad}^J_{\sf X}{\sf Y}
.\end{eqnarray*}
\end{proof}
\noindent
Equation (\ref{xder-rel}) provides the dynamics of the quantity ${\sf Y}({s},t)$.
Lemma \ref{lemma-JB} and Lemma \ref{lemma-xder} now allow us to prove the following theorem for the derivation of the EP equation for $Sp(2)$.
\begin{theorem}
Hamilton's principle $\delta S=0$ for
$
S=\int_a^b \ell({\mathsf{X}},{\mathsf{Y}})\,ds\,dt
$
implies the EP equation for the symplectic Lie algebra $\mathfrak{sp}(2)$,
\begin{equation}
\frac{\partial}{\partial t}\frac{\delta\ell}{\delta {\mathsf{X}}}
+ 2 {\rm sym}\Big(J{\sf X}\frac{\partial\ell}{\partial {\sf X}}\Big)
 + \frac{\partial}{\partial {s}} \frac{\delta\ell}{\delta {\mathsf{Y}}}
+2 {\rm sym}\Big(J{\sf Y}\frac{\partial\ell}{\partial {\sf Y}}\Big)
=0
\,.
\label{EPsymp}
\end{equation}
\end{theorem}

\begin{proof}
By a direct calculation, integrating by parts, substituting and rearranging yields
\begin{eqnarray*}
\delta S
\!\!\!&=&\!\!\!
\int_a^b
\Big\langle \frac{\delta\ell}{\delta {\mathsf{X}}}
\,,\,\delta {\mathsf{X}} \Big\rangle
+
\Big\langle \frac{\delta\ell}{\delta {\mathsf{Y}}}
\,,\,\delta {\mathsf{Y}} \Big\rangle
\,ds\,dt
\\
\!\!\!&=&\!\!\!
\int_a^b
\Big\langle \frac{\delta\ell}{\delta {\mathsf{X}}}
\,,\,\partial_t {\mathsf{Z}}
+ {\rm ad}^J_{\mathsf{X}}{\mathsf{Z}}\Big\rangle
+
\Big\langle \frac{\delta\ell}{\delta {\mathsf{Y}}}
\,,\,\partial_{s} {\mathsf{Z}}
+ {\rm ad}^J_{\mathsf{Y}}{\mathsf{Z}} \Big\rangle
\,ds\,dt
\\
\!\!\!&=&\!\!\!
\int_a^b\!
\Big\langle \!-\partial_t \frac{\delta\ell}{\delta {\mathsf{X}}}
+ {\rm ad}^{J*}_{\mathsf{X}}\frac{\delta\ell}{\delta {\mathsf{X}}}
\,,\,{\mathsf{Z}}\Big\rangle
+
\Big\langle\! -\partial_{s} \frac{\delta\ell}{\delta {\mathsf{Y}}}
+ {\rm ad}^{J*}_{\mathsf{Y}}\frac{\delta\ell}{\delta {\mathsf{Y}}}
\,,\,{\mathsf{Z}} \Big\rangle
\,ds\,dt
\\
\!\!\!&=&\!\!\!
\int_a^b
\Big\langle -\frac{\partial}{\partial t}\frac{\delta\ell}{\delta {\mathsf{X}}}
+ {\rm ad}^{J*}_{\mathsf{X}}\frac{\delta\ell}{\delta {\mathsf{X}}}
- \frac{\partial}{\partial {s}} \frac{\delta\ell}{\delta {\mathsf{Y}}}
+ {\rm ad}^{J*}_{\mathsf{Y}}\frac{\delta\ell}{\delta {\mathsf{Y}}}
\,,\,{\mathsf{Z}} \Big\rangle
\,ds\,dt
\\
\!\!\!&=&\!\!\!
- \int_a^b
\Big\langle \frac{\partial}{\partial t}\frac{\delta\ell}{\delta {\mathsf{X}}}
+ 2 {\rm sym}\Big(J{\sf X}\frac{\partial\ell}{\partial {\sf X}}\Big)
 + \frac{\partial}{\partial {s}} \frac{\delta\ell}{\delta {\mathsf{Y}}}
+2 {\rm sym}\Big(J{\sf Y}\frac{\partial\ell}{\partial {\sf Y}}\Big)
\,,\,{\mathsf{Z}} \Big\rangle
\,ds\,dt
\,.
\end{eqnarray*}
Here the formulas for the variations $\delta {\mathsf{X}}$ and $\delta {\mathsf{Y}}$ are obtained from the usual equality of cross derivatives. One defines the $J$-coadjoint action $(\ad^{J*})$ as the dual of the $J$-adjoint action $(\ad^J)$ with respect to the trace pairing $\langle\,\cdot\,,\,\cdot\,\rangle$ of symmetric matrices by
\begin{equation}
\Big\langle
{\rm ad}^{J*}_{\mathsf{X}}\frac{\delta\ell}{\delta {\mathsf{X}}}
\,,\,{\mathsf{Z}} \Big\rangle
=
\Big\langle \frac{\delta\ell}{\delta {\mathsf{X}}}
\,,\,{\rm ad}^J_{\mathsf{X}}{\mathsf{Z}} \Big\rangle
=
\Big\langle \frac{\delta\ell}{\delta {\mathsf{X}}}
\,,\,2{\rm sym}({\sf X}J{\sf Z})\Big\rangle
=
\Big\langle - 2 {\rm sym}\Big(J{\sf X}\frac{\partial\ell}{\partial {\sf X}}\Big)
\,,\,{\sf Z}\Big\rangle
\,.
\label{ad-star-def}
\end{equation}
\end{proof}

\begin{remark}
Consider the evolution equation (\ref{EPsymp}) in the case that
\begin{equation}
\ell({\sf X}, {\sf Y})
= \frac{1}{2} \int {\rm tr}\big({\sf X}^2\big) - {\rm tr}\big({\sf Y}^2\big)\,d{s}
\,,
\label{sp-chainEP}
\end{equation}
where ${\rm tr}$ denotes trace of a matrix. (This is the motion equation for the $Sp(2)$ chain with respect to the trace pairing of symmetric matrices.)

For this Lagrangian, we have $\partial\ell/\partial {\sf X}={\sf X}$ and $\partial\ell/\partial {\sf Y}=-\,{\sf Y}$, so the Euler-Poincar\'e equation (\ref{EPsymp}) becomes
\begin{equation}
\dot{{\sf X}}
+ 2 {\rm sym}\left(J{\sf X}^2\right)
= {\sf Y}'
+ 2 {\rm sym}\left(J{\sf Y}^2\right)
,
\label{EPsymp1}
\end{equation}
or, equivalently, in terms of the ordinary commutator
\begin{equation}
\dot{{\sf X}}
+ [{\sf X}, {\sf X}J + J {\sf X} ]
=  {\sf Y}'
+ [{\sf Y}, {\sf Y}J + J {\sf Y} ]
\,.
\label{BIchain}
\end{equation}
\end{remark}

\begin{remark}
When ${\sf Y}$ is absent, equation (\ref{BIchain}) for the BI strand recovers the BI ordinary differential equation (\ref{BI-eqn2}) for the case that $N=J$. 
\end{remark}

\begin{remark}
Either of the equivalent forms (\ref{EPsymp1}) or (\ref{BIchain}) of the EP(symp) motion equation must be closed by using the evolution equation for ${\sf Y}$, obtained in equation (\ref{xder-rel}) from the compatibility of cross derivatives. Namely,
\begin{equation}
\dot{{\sf Y}} =  {\sf X}'  - 2 {\rm sym}\left({\sf X}J{\sf Y}\right)
,
\label{compat-eqn1}
\end{equation}
or, equivalently,
\begin{equation}
\dot{{\sf Y}} =  {\sf X}'  - ({\sf X}J{\sf Y} - {\sf Y}J{\sf X})
\,.
\label{compat-eqn2}
\end{equation}

Equations (\ref{BIchain}) and (\ref{compat-eqn2}) may be compared against corresponding equations  (\ref{XY-eqn-2}) and (\ref{XY-eqn-1}), respectively.
In this comparison, these systems seem to be dual to each other, as symmetric is dual to antisymmetric.
\end{remark}

\subsubsection{Lie-Poisson Hamiltonian form of $Sp(2)$ strand dynamics}
The Hamiltonian form of the Euler-Poincar\'e equation (\ref{BIchain}) and compatibility equation (\ref{xder-rel}) with its associated Lie-Poisson bracket may be found as before from the Legendre transform
\[
 {\sf P} = \frac{\partial\ell}{\partial {\sf X}}
\quad\hbox{and}\quad
{\sf X} = \frac{\partial  h}{\partial {\sf P}}
\quad\hbox{with}\quad
h({\sf P},{\sf Y})={\rm tr}({\sf P} {\sf X}) - \ell({\sf X},{\sf Y})
\,.
\]
The Lie-Poisson bracket may also be obtained in the usual way as
\begin{equation} \label{LP-Ham-struct-sp2}
\frac{\partial}{\partial t}
    \begin{bmatrix}
    {\mathsf{P}}
    \\
    {\mathsf{Y}}
    \end{bmatrix}
=
\begin{bmatrix}
    {\rm ad}^{J*}_{\Box}{\mathsf{P}}
   &    
       \partial_{s}
   - 
   {\rm ad}^{J*}_{\mathsf{Y}}
      \\
    \partial_{s}
   +
   {\rm ad}^{J}_{\mathsf{Y}}
& 0
    \end{bmatrix}
    \begin{bmatrix}
   \delta h/\delta{\mathsf{P}} \\
   \delta h/\delta{\mathsf{Y}}
    \end{bmatrix}.
\end{equation}
This is the $J$-bracket version of the Lie-Poisson Hamiltonian structure in (\ref{LP-Ham-struct-symbol}); so 
its Jacobi identity will follow from that of the J-bracket discussed earlier.

Here, as in (\ref{ad-star-def}), one defines the $J$-coadjoint action $(\ad^{J*})$ as the dual of the $J$-adjoint action $(\ad^J)$ with respect to the trace pairing $\langle\,\cdot\,,\,\cdot\,\rangle$ of symmetric matrices. Explicitly, this is 
\begin{equation}
\Big\langle
{\rm ad}^{J*}_{{\partial  h}/{\partial {\sf P}}}{\mathsf{P}}
\,,\,{\mathsf{Z}} \Big\rangle
=
\Big\langle {\mathsf{P}}
\,,\,{\rm ad}^J_ {{\partial  h}/{\partial {\sf P}}}{\mathsf{Z}} \Big\rangle
=
\Big\langle {\mathsf{P}}
\,,\,2{\rm sym}\Big( \frac{\partial  h}{\partial {\sf P}}J{\sf Z}\Big)\Big\rangle
=
\Big\langle - 2 {\rm sym}\Big(J \frac{\partial  h}{\partial {\sf P}}{\mathsf{P}}\Big)
\,,\,{\sf Z}\Big\rangle
\,.
\label{ad-star-def}
\end{equation}

\begin{remark}
Legendre transforming the Lagrangian (\ref{sp-chainEP}) that is a difference of squares  leads to a Hamiltonian that is a sum of squares
\begin{equation}
h({\sf P}, {\sf Y})
= \frac{1}{2} \int {\rm tr}({\sf P}^2) + {\rm tr}({\sf Y}^2)\,d{s}
\,,
\label{sp-chainHam}
\end{equation}
where, as before, ${\rm tr}$ denotes the trace of a matrix. 
\end{remark}

\subsubsection{The isomorphism between $Sp(2)$ and $SL(2)$ strand equations}
The Lie algebras $\mathfrak{sp}(2)$ and $\mathfrak{sl}(2)$ are isomorphic.  The $G$-strand equations in terms of the variables ${\sf n}\in \mathfrak{sl}(2)$ and ${\sf m}\in \mathfrak{sl}(2)$ are
\begin{eqnarray}
\dot{{\sf m}}-{\sf n}'-[{\sf m}^{T},{\sf m}]+[{\sf n}^T,{\sf n}]=0,\qquad \text{(EP equation}) 
\label{sl1}
\\
\dot{{\sf n}}-{\sf m}'+[{\sf n},{\sf m}]=0 \qquad \text{(compatibility)}, 
\label{sl2}
\end{eqnarray}
in which $J {\sf m}$ and $J {\sf n}$ are symmetric. The first of these is the $ \mathfrak{sl}(2)$ EP equation with Lagrangian
\begin{equation}
\mathcal{L}=\frac{1}{2}\int\text{tr} ({\sf n}^T {\sf n} + {\sf m}^T {\sf m}) d{s} \,.
\label{XY-eqn-1-so3-tld}
\end{equation}
Equations (\ref{sl1}) and (\ref{sl2}) are related to the previously considered $K$-chiral model (that is, before specifying how $K$ relates to $P$) with 
\begin{eqnarray}
{\sf K}
=
\left(
\begin{array}{ccc}
0 & 1 & 0 \\
1 & 0 & 0 \\
0 & 0 & 1
\end{array}
\right),
\label{K-nondiag}
\end{eqnarray} 
where however this ${\sf K}$ is not diagonal. 
The relation is based on the following linear invertible map between vectors $\boldsymbol{\mathsf{X}}\in\mathbb{R}^3$ and $\mathfrak{sl}(2)$ matrices:
\begin{eqnarray}
{\sf A}(\boldsymbol{\mathsf{X}})
=
\left(
\begin{array}{ccc}
\frac{X_3}{2} & \frac{X1}{\sqrt{2}}  \\
 \frac{X_2}{\sqrt{2}} & -\frac{X_3}{2}  
\end{array}
\right) \in \mathfrak{sl}(2).
\label{Arep}
\end{eqnarray} 
Consequently,  one finds 
\begin{eqnarray*}
[{\sf A}(\boldsymbol{\mathsf{X}}),{\sf A}(\boldsymbol{\mathsf{Y}})]&=&
{\sf A}(\mathsf{K}\cdot\boldsymbol{\mathsf{X\times Y}})
\\
{\sf A}({\sf K}\cdot \boldsymbol{\mathsf{X}})&=&{\sf A}^T(\boldsymbol{\mathsf{X}}) 
\end{eqnarray*}and
\begin{eqnarray*}
[{\sf A}(\boldsymbol{\mathsf{X}}),{\sf A}(\boldsymbol{\mathsf{Y}})]
&=&
{\sf A}^T(\boldsymbol{\mathsf{X\times Y}})\\
{\sf A}(\boldsymbol{\mathsf{X\times Y}})&=&([{\sf A}(\boldsymbol{\mathsf{X}}),{\sf A}(\boldsymbol{\mathsf{Y}})])^T
\,.\end{eqnarray*} 
Thus, the constraint equation for the $K$-model (\ref{XY-eqn-1}) in the $\mathfrak{sl}(2)$ representation becomes
\begin{eqnarray}
\dot{{\sf A}}(\boldsymbol{\mathsf{Y}})-{\sf A}'(\boldsymbol{\mathsf{X}})=-{\sf A}(\mathsf{K}\cdot\boldsymbol{\mathsf{X\times Y}})=-[{\sf A}(\boldsymbol{\mathsf{X}}),{\sf A}(\boldsymbol{\mathsf{Y}})]. 
\end{eqnarray} 
Likewise, the EP equation for the $K$-model (\ref{XY-eqn-2}) becomes 
\begin{eqnarray}
{\sf A}'(\boldsymbol{\mathsf{Y}})-\dot{{\sf A}}(\boldsymbol{\mathsf{X}})
&=& -{\sf A}(\boldsymbol{\mathsf{Y}}\times {\sf K}\cdot \boldsymbol{\mathsf{Y}})+A(\boldsymbol{\mathsf{X}}\times {\sf K}\cdot \boldsymbol{\mathsf{X}})
\nonumber\\
&=&
-[{\sf A}(\boldsymbol{\mathsf{Y}}),{\sf A}({\sf K}\boldsymbol{\mathsf{Y}})]^T + [{\sf A}(\boldsymbol{\mathsf{X}}),{\sf A}({\sf K} \boldsymbol{\mathsf{X}})]^T 
\nonumber\\
&=&
-[{\sf A}(\boldsymbol{\mathsf{Y}}),{\sf A}^T(\boldsymbol{\mathsf{Y}})]^T + [{\sf A}(\boldsymbol{\mathsf{X}}),{\sf A}^T(\boldsymbol{\mathsf{X}})]^T 
\nonumber\\
&=&
-[{\sf A}(\boldsymbol{\mathsf{Y}}),{\sf A}^T(\boldsymbol{\mathsf{Y}})] + [{\sf A}(\boldsymbol{\mathsf{X}}),{\sf A}^T(\boldsymbol{\mathsf{X}})] 
\,.
\end{eqnarray} 

Now upon identifying ${\sf m}=\mathsf{A}({\boldsymbol{\mathsf{X}}})$ and ${\sf n}=\mathsf{A}({\boldsymbol{\mathsf{Y}}})$ one finds that the $\mathfrak{sl}(2)$ model equations (\ref{sl1})--(\ref{sl2}) above are also integrable, because they are isomorphic to (\ref{XY-eqn-1}) - (\ref{XY-eqn-2}).

Furthermore, one can make use of the following property for the action of
${\sf K}$ given in (\ref{K-nondiag}): the vector $ {\sf K}\cdot \boldsymbol{\mathsf{X}}\in \mathbb{R}^3$ corresponds to the matrix $-{\sf K}{\sf X}{\sf K} \in \mathfrak{so}(3)$,
where as usual ${\sf X}$ is the $\mathfrak{so}(3)$ representation of the vector $\boldsymbol{\mathsf{X}}$.
The equations of the $\mathfrak{sp}(2)$ model according to the above $\mathfrak{sl}(2)$ identification can be written in vector form (\ref{XY-eqn-1}) - (\ref{XY-eqn-2}). In the $\mathfrak{so}(3)$ representation, this becomes
\begin{equation}
\dot{{\sf X}}
+ [{\sf K}{\sf X}{\sf K}, {\sf X}]
=  {\sf Y}'
+ [{\sf K}{\sf Y}{\sf K}, {\sf Y}]
\qquad\hbox{and}\qquad
\dot{{\sf Y}} =  {\sf X}'  - {\sf K}[{\sf Y},{\sf X}]{\sf K}
\,.
\label{XY-eqn-1-so3}
\end{equation}
The symmetric matrix ${\sf K}$ may be written as ${\sf K}={\sf O}{\sf \tilde{K}}{\sf O^T}$, with diagonal matrix ${\sf \tilde{K}}=\text{diag}(-1,1,1)$ and orthogonal matrix ${\sf O}$ given by
\begin{eqnarray}
{\sf O}
=
\left(
\begin{array}{ccc}
-\frac{1}{\sqrt{2}} &\frac{1}{\sqrt{2}} & 0 \\
\frac{1}{\sqrt{2}} & \frac{1}{\sqrt{2}} & 0 \\
0 & 0 & 1
\end{array}
\right),
\quad \text{note ${\sf O}={\sf O}^T={\sf O}^{-1}$}. 
\label{O-nondiag}
\end{eqnarray} 
In terms of the new variables ${\sf \tilde{X}}={\sf O}^T{\sf X}{\sf O}$,
${\sf \tilde{Y}}={\sf O}^T{\sf Y}{\sf O}$, equations (\ref{XY-eqn-1-so3})
keep their form 
\begin{equation}
\dot{\tilde{\sf X}}
+ [\tilde{\sf K}\tilde{\sf X}{\sf K},\tilde {\sf X}]
=  \tilde{\sf Y}'
+ [\tilde{\sf K}\tilde{\sf Y}\tilde{\sf K}, \tilde{\sf Y}]
\qquad\hbox{and}\qquad
\dot{\tilde{\sf Y}} =  \tilde{\sf X}'  - \tilde{\sf K}[\tilde{\sf Y},\tilde{\sf X}]\tilde {\sf K}
\,.
\label{XY-eqn-1-so3-tld}
\end{equation}
However, now the matrix $ \tilde{\sf K}$ is diagonal and this case is integrable,
since it reduces to the P-chiral model.
Thus, with the `gauge' transformation that diagonalizes  ${\sf K}$ and a series of canonical isomorphisms, the $\mathfrak{sp}(2)$ EP equations with Lagrangian (\ref{XY-eqn-1-so3-tld}) have also been brought into the form of the integrable $\mathfrak{so}(3)$ P-chiral model.

\section{Numerical solutions for $O(3)$ and $Sp(2)$ $P$-Chiral models}\label{JRPnumerics}

\newcommand{\bmu}{\boldsymbol{\mathsf{u}}}
\newcommand{\bmv}{\boldsymbol{\mathsf{v}}}

The equations for the $O(3)$ $P$-Chiral model (\ref{uv-eqn-1}--\ref{uv-eqn-2}) can be rewritten in the form

\begin{equation}
\label{eq_u_xi}
\partial_t \bmu(t,\xi) +\bmu(t,\xi)\times P\bmv(t,\xi)=0 \,, 
\end{equation}

\begin{equation}
\label{eq_v_eta}
\partial_t \bmv(t,\eta) +\bmv(t,\eta)\times P\bmu(t,\eta)=0 \,,
\end{equation}
for new wave-tracking variables
\[ \xi = {s}-t\,, \qquad \eta = {s}+t\,.\]
This formalism shows explicity the nature of the cross-coupling between the natural physical variables, $\bmu$ and $\bmv$. At \emph{fixed} $\xi$ (i.e. along a characteristic of fixed slope, $c=1$), $\bmu$ may vary only when it interacts with nonvanishing $\bmv$. Similarly, at fixed $\eta$ (a characteristic of fixed slope $c=-1$), $\bmv$ varies only with nonvanishing $\bmu$. Note also, as previously stated, the quantities $|\bmu|(\xi)$ and $|\bmv|(\eta)$ are fixed along their respective characteristics. This in turn implies constraints on the evolution of the original geometric varibles $X=(\bmu-\bmv)/2$ and $-(\bmv+\bmu)/2$. 

In the static ${s}$ frame, compactly supported solutions of the $G$-strand models consist of trains of generally fixed travelling profiles in $\bmu$ and $\bmv$ along the line at speed $c=\pm 1$ and undergoing rotation when they collide, or equivalently matched symmetric and antisymmetic pairs in $X$,$Y$ undergoing the same behaviour.

The model (\ref{eq_u_xi}--\ref{eq_v_eta}) is particularly tractable for numerical simulation since the absence of explicit spatial derivatives removes the numerical dispersion typical in time integration of advective operators such as appear in (\ref{uv-eqn-1}--\ref{uv-eqn-2}). The forcings due to of the nonlinear interaction terms can be modelled through standard time discretization techniques. Here we choose to use the implicit midpoint rule using cubic spline interpolation over a pointwise discretization of $\xi$ and $\eta$. 

Figure \ref{fig_waterfall_1} shows waterfall plot comparisons of the three components of the numerical solution solution $\boldsymbol{\mathsf{u}}$ and of $\boldsymbol{\mathsf{v}}$ for the standard Chiral equations, with matrix
\[ K = I, \]
and hence, from (\ref{def-P})
\[ P = I.\]
The equations are integrated on a periodic domain in $s$ with initial conditions
\begin{eqnarray}
u(0,{s})=(0,A({s}-0,25),0),&&\mbox{for }{s} \in [-0.25,0.75]
\nonumber\\ 
v(0,{s})= (A({s}-0.75),0,0),&& \mbox{for }{s} \in [0.25,1.25] 
\label{fig1-ic}\\ 
\mbox{for a Gaussian profile} &&A({s})=\exp(-256{s}^2)
\nonumber
\end{eqnarray}
This choice of profile for the experiment is effectively arbitrary and in fact any suitably smooth profile could be substituted. The observed solution behaviour is as expected, with rightward, $c=1$, travelling waves in $\bmu$ and leftward, $c=-1$, waves in $\bmv$. At points of collision the nonlinear terms cause the waves counterrotate around each other, primarily into the third component due to the value of $\bmu\times\bmv$ at $t=0$.

Following the methodology and isomorphisms of the previous section we may also apply the same numerical model and analysis for the BI strand equations (\ref{BIchain}) by specifying matrices
\[K=\left(
\begin{array}{ccc}
-1& 0& 0\\
 0& 1& 0\\
 0& 0& 1
\end{array}\right),
\quad 
P=\frac{1}{2}
\left(
\begin{array}{ccc}
 3& 0& 0\\
 0&-1& 0\\
 0& 0&-1
\end{array}\right)\,,
\]
and setting the vectors $\bmu$, $\bmv$ to map to and from symmetric matrices $X$, $Y$ through
\[\bmu=(u_1,u_2,u_3)=
\frac{1}{2}
\left(\begin{array}{cc}
u_1+u_2&u_3\\
u_3&u_1-u_2
\end{array}\right)=X-Y
\]
\[\bmv=(v_1,v_2,v_3)=
\frac{1}{2}
\left(\begin{array}{cc}
v_1+v_2&v_3\\
v_3&v_1-v_2
\end{array}\right)=-X-Y
\]
Figure \ref{fig_waterfall_2} shows the results of an integration in the BI mode of the vector initial conditions (\ref{fig1-ic}).

\paragraph{Discussion.}
The numerical results exhibit fundamentally the same behaviour, except that the isotropy breaking in $P$  introduces preferred axes to the rotation so that transfers between the second and third components of $\bmu$ and $\bmv$ are faster than motions involving the first component. Hence, in this case, we observe a rapid growth in $u_3$ in comparison to $v_3$ over the course each collision. It is in part this directionality in the rotations which allows the Bloch-Iserles ordinary differential equations (ODE) to act a sorting process.

\begin{figure}[tbp]

\includegraphics[width=\textwidth]{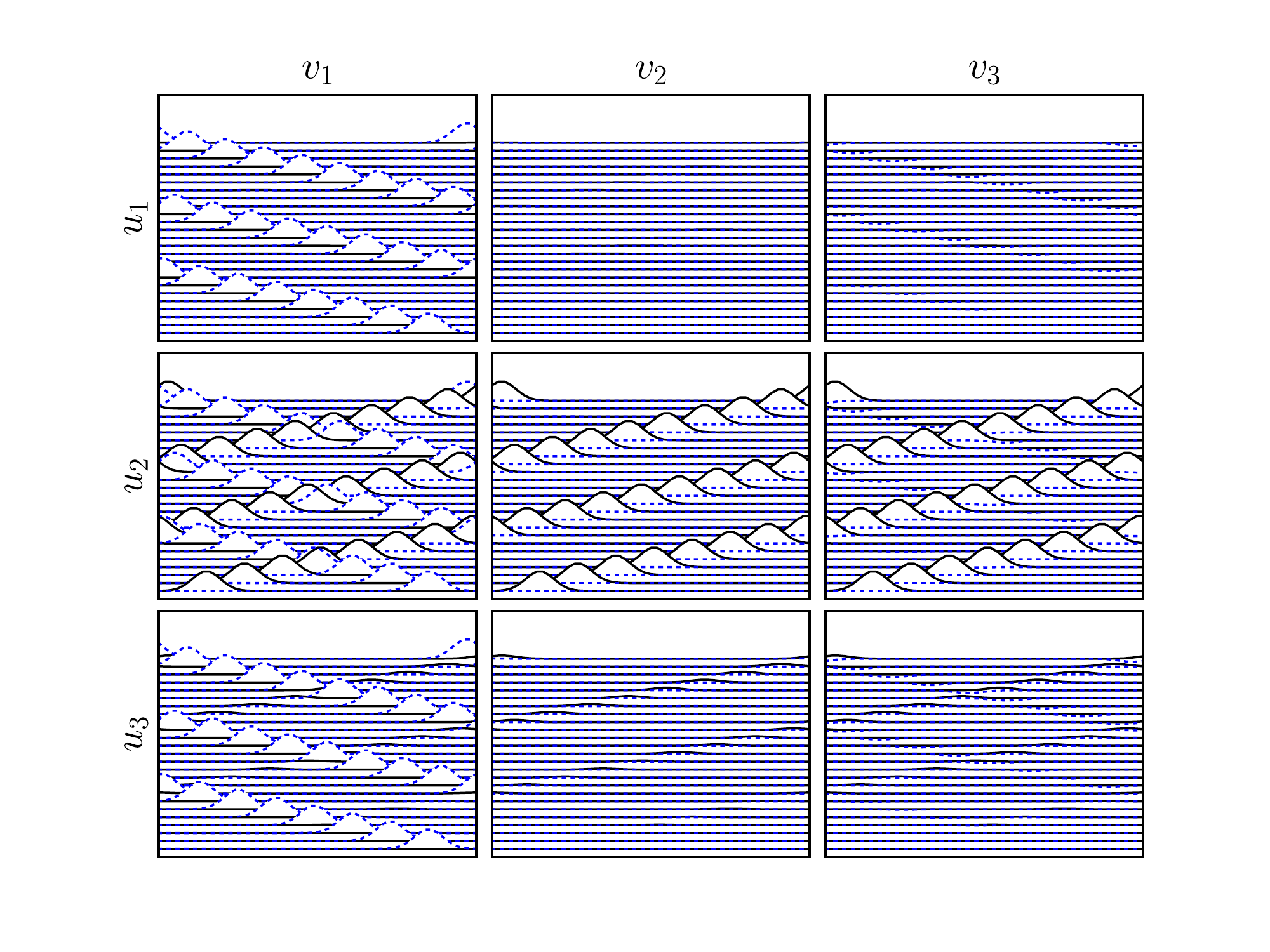}

\caption{\sf
\label{fig_waterfall_1} 
Array of waterfall plots showing the $t$ and $s$ dependence and componentwise comparison of P-Chiral variables $\bm{u}$ and $\bm{v}$ for the numerical solution of the initial value problem (\ref{fig1-ic}) in a periodic domain of unit length. Here the components of $\bm{u}$ are plotted in solid black in rows down the page and the components of $\bm{v}$ in dashed blue in columns across the page. Snapshots of profiles in $s$ are shown at various times $t$ with the vertical offset proportional to $t$.  In the absence of the signal in the other variable each profile would translate forever at unit wavespeed, rightwards for $\bm{u}$, leftwards for $\bm{v}$. This motion would be independent of the actual profile chosen. At points of interaction the variables are rotated antisymmetrically, introducing a clear signal in the cross-component direction, $u_3$, $v_3$ as well as a much less obvious response in $u_1$ and $v_2$. The magnitudes $|\bm{u}|$ and $|\bm{v}|$ are conserved in frames travelling with the wave, even through the collision.
}
\end{figure}

\begin{figure}[tbp]

\includegraphics[width=\textwidth]{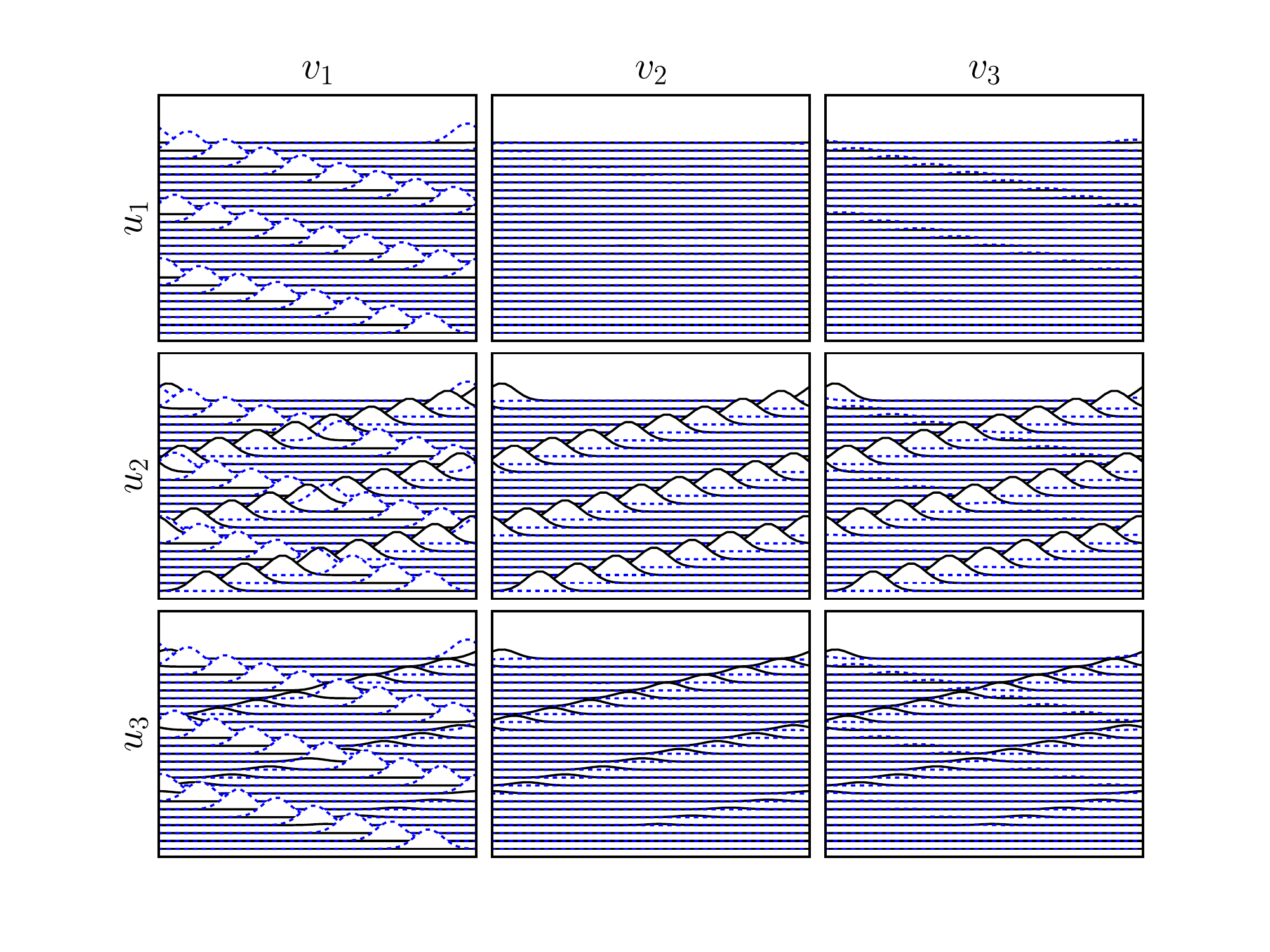}

\caption{\sf
\label{fig_waterfall_2} 
Array of waterfall plots showing the $t$ and $s$ dependence and componentwise comparison of Bloch-Iserles variables $\bm{u}$ and $\bm{v}$ for the numerical solution of the initial value problem (\ref{fig1-ic})
in a periodic domain of unit length. Here the components of $\bm{u}$ are plotted in solid black in rows down the page and the components of $\bm{v}$ in dashed blue in columns across the page. Snapshots of profiles in $s$ are shown at various times $t$ with the vertical offset proportional to $t$. In the absence of the signal in the other variable each profile would translate forever at unit wave speed, rightwards for $\bm{u}$, leftwards for $\bm{v}$. Compared with the previous figure, the much larger signal in $u_3$ post collision is obvious. This is due to the anisotropic nature of the Bloch-Iserles strand operator. The magnitudes $|\bm{u}|$ and $|\bm{v}|$ are again conserved in each  travelling frame.}
\end{figure}

Given the sorting nature of the Bloch-Iserles ODE (when $Y=0$ and $X$ is independent of ${s}$), it is natural to question whether this behaviour is observed on the spatially dependent strand. To address this question, we now consider the evolution  of the initial conditions for two forms of twisted helical strand, namely
\[
X_1=\frac{1}{2}\left(\begin{array}{cc}
\sin(2 \pi {s})+\cos(2 \pi {s})&0\\
0&\sin(2 \pi {s})-\cos(2 \pi {s})
\end{array}\right),\quad 
Y_1=\frac{1}{2}\left(\begin{array}{cc}
0&0\\
0&0
\end{array}\right)\,,\]
\[
X_2=\left(\begin{array}{cc}
\sin(2 \pi {s})&\cos(2\pi {s})\\
\cos(2\pi {s})&-\sin(2 \pi {s})
\end{array}\right),\quad 
Y_2=\frac{1}{2}\left(\begin{array}{cc}
0&0\\
0&0
\end{array}\right)\,,\]
i.e. under the mapping to vector space, 
\[ \bm{X_1}=(\sin(2 \pi {s}),\cos(2 \pi {s}),0),\qquad \bm{Y_1}=\bm{0}, \]
\[ \bm{X_2}=(0,\sin(2 \pi {s}),\cos(2 \pi {s})),\qquad \bm{Y_2}=\bm{0}, \]
or, in the P-Chiral physical variables,
\[ \bmu_1=(\sin(2 \pi {s}),\cos(2 \pi {s}),0),\qquad \bmv_1=(-\sin(2 \pi {s}),-\cos(2 \pi {s}),0). \]
\[ \bmu_2=(0,\sin(2 \pi {s}),\cos(2 \pi {s})),\qquad \bmv_2=(0,-\sin(2 \pi {s}),-\cos(2 \pi {s})). \]
These are two of  the simplest (in the spectral sense) states which can exhibits spatial dependence. In $\bmu$--$\bmv$ space the solutions would consist, in the absence of the other component of counter-rotation helices. Figures \ref{fig_eigen1} and \ref{fig_eigen2} show the variation of locus of the eigenvectors of the various matrices with time. In both cases it will be noted that the oscillatory behaviour caused by the wave motion dominates, in this sense motions on the strand retard the sorting action of the BI operator as the algorithm can only operate when characteristics in $\bmu$ and $\bmv$ originating from the same point in $S^1$ coincide. 

\begin{figure}[tbp]  
  \includegraphics[width=0.95\textwidth]{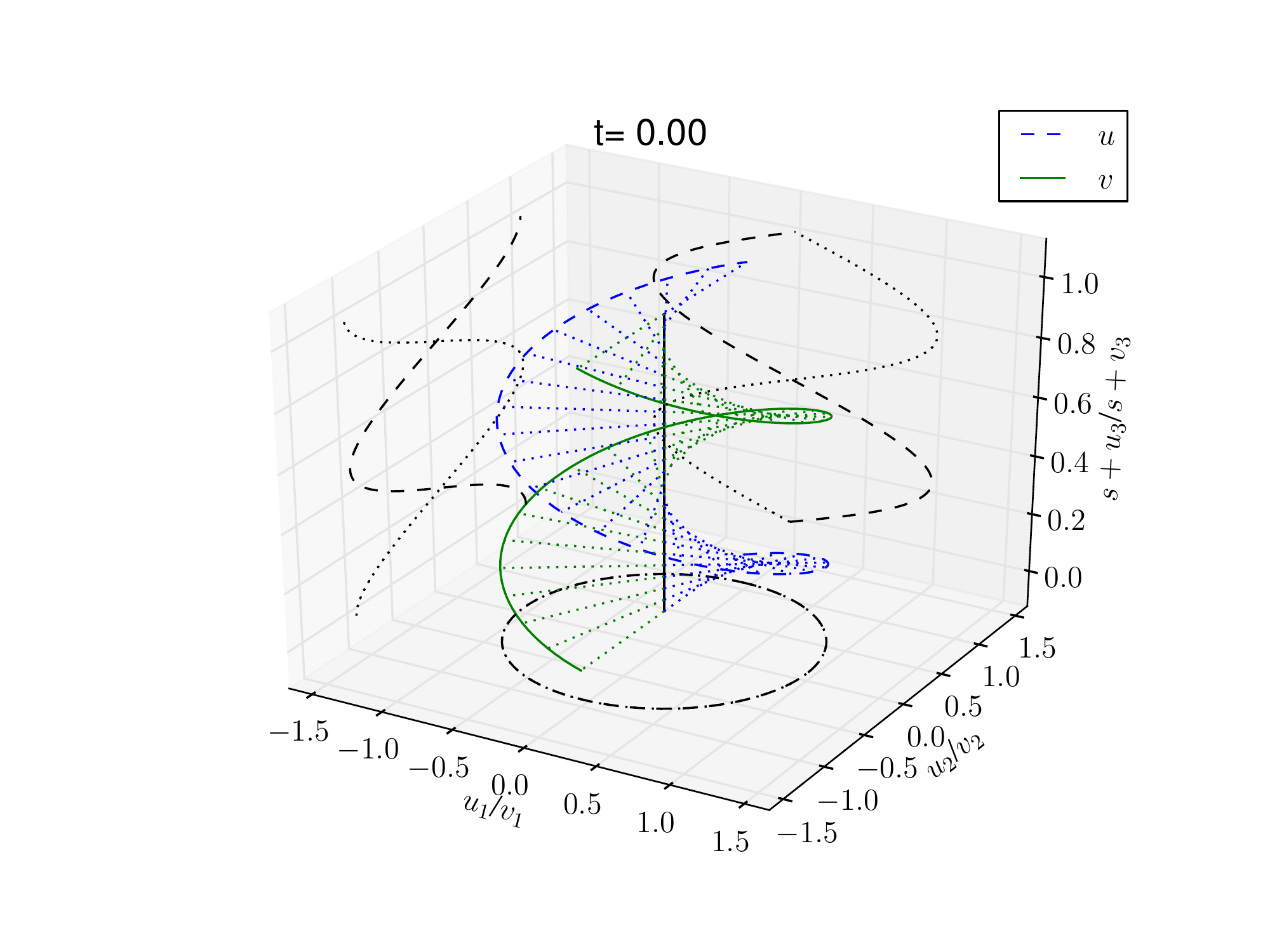}
\caption{\label{figChiralICs}\sf
Isometric plots of the helical initial conditions used in the Bloch-Iserles strand problem shown in terms of the $\bmu_1$,$\bmv_1$ variables.
In the first panel one period of the spatial domain, ${s}$ is mapped to the vertical axis, with the $\mathbb{R}^3$ vectors $\bmu_1({s})$ and $\bmv_1({s})$ plotted at 25 locations (dotted blue and green lines respectively), linking the point ${s}$ with points ${s}+\bmu_1({s})$ and ${s}+\bmv_1({s})$. The locus of these points is indicated in the solid lines. The projections of these loci onto the ${s}$-$y$, ${s}$-$z$, and $y$-$z$ planes are also indicated in dashed black for $\bmu$ and dotted black for $\bmv$. The initial condition for $\bmu_2$, $\bmv_2$ has fundamentally the same form, oriented albeit rotated to twist around the ${s}$-axis.
Animations of the evolution of these visualizations are available as supplementary materials for numerical experiments presented in this paper.}
\end{figure}

\begin{figure}[tbp]
{\centering
  (a) \includegraphics[width=0.47\textwidth]{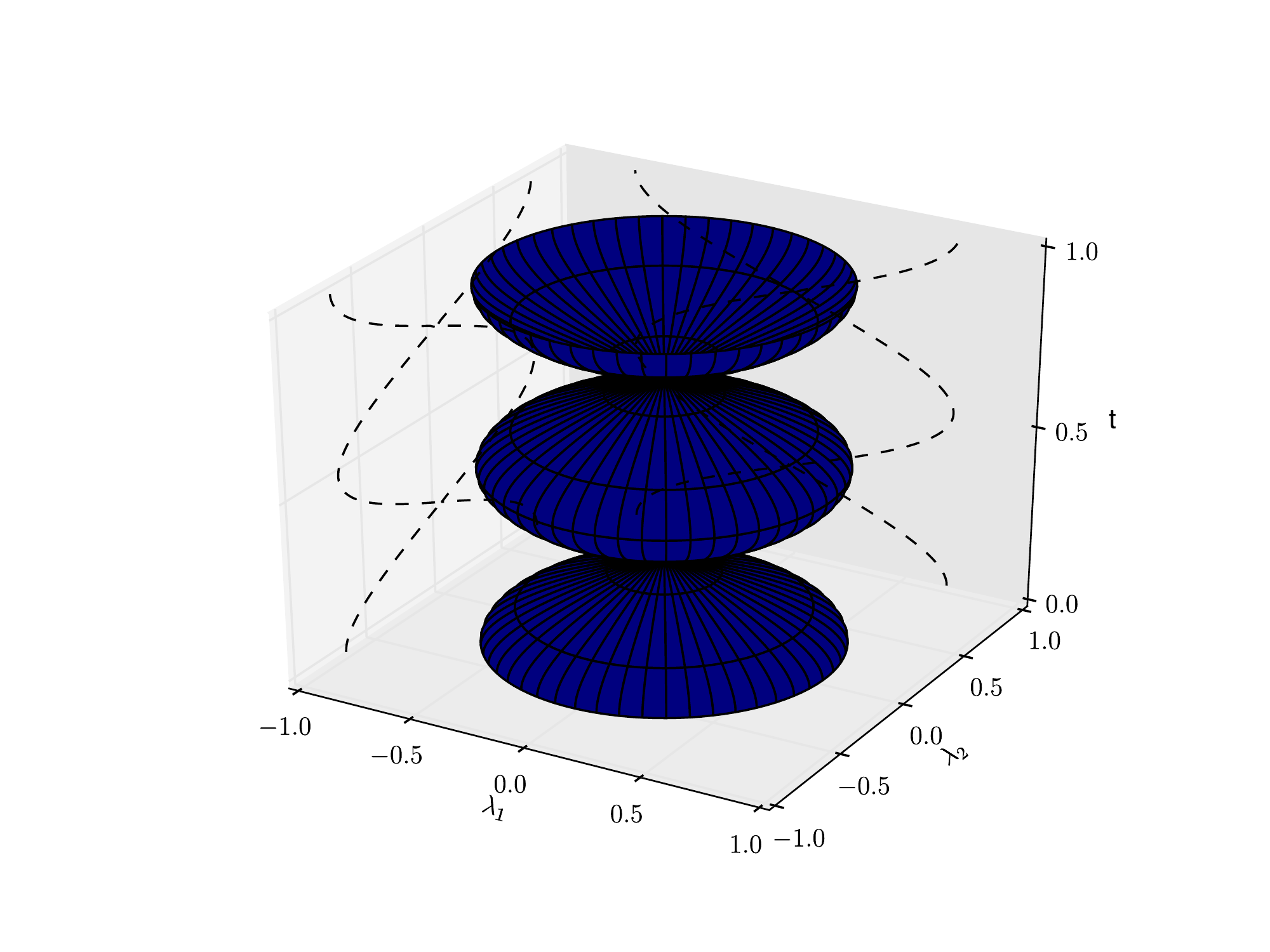}
  (b) \includegraphics[width=0.47\textwidth]{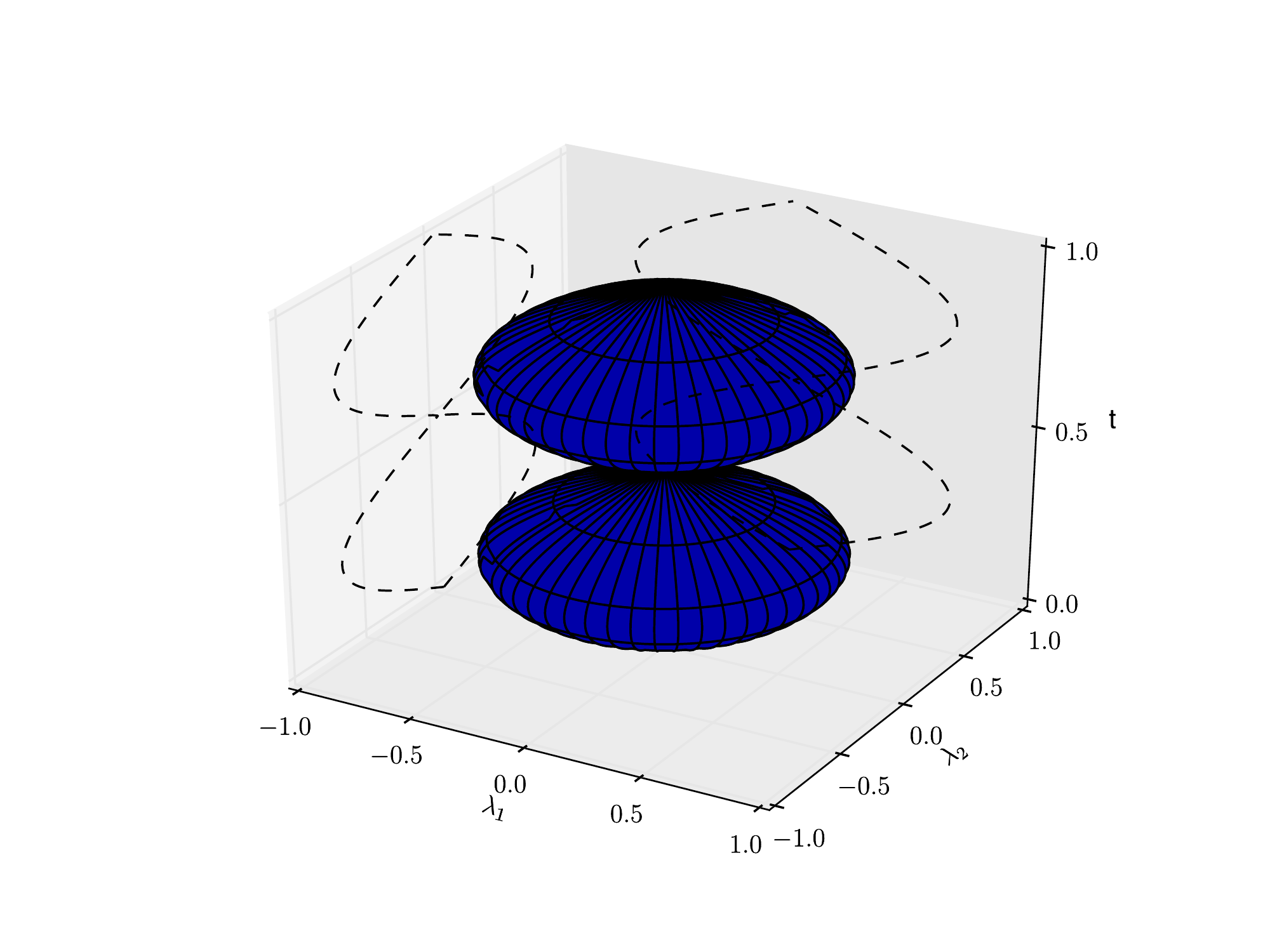}
}

\caption{\sf \label{fig_eigen1}
Isometric plots  of surfaces generated by the eigenvalues of the symmetric matrices $X({s},t)$ (Panel a) and $Y({s},t)$ (Panel b) under the  strand Bloch-Iserles equations for helical initial conditions,
\newline  \newline
\centerline{
$\bm{u}=(\sin(2\pi {s}),\cos(2\pi {s}),0),\quad\bm{v}=(\sin(2\pi {s}),\cos(2\pi {s}),0), \ \ {s}\in [0,1],
$
}
\newline\newline
on a periodic domain. Time dependence is plotted along the vertical axis, with eigenvectors $\lambda_1$,$\lambda_2$ on the horizontal axes. Linear oscillatory motions dominate, with the locus of eigenvectors lying close to a circle, whose radius varies sinusoidally with time in both variables, with a clear half period phase difference. Animations of the time evolution of these solutions are available as supplementary materials for all numerical experiments presented in this paper.} 
\label{figBIlambda1}
\end{figure}

\begin{figure}[tbp]
{\centering
  (a) \includegraphics[width=0.47\textwidth]{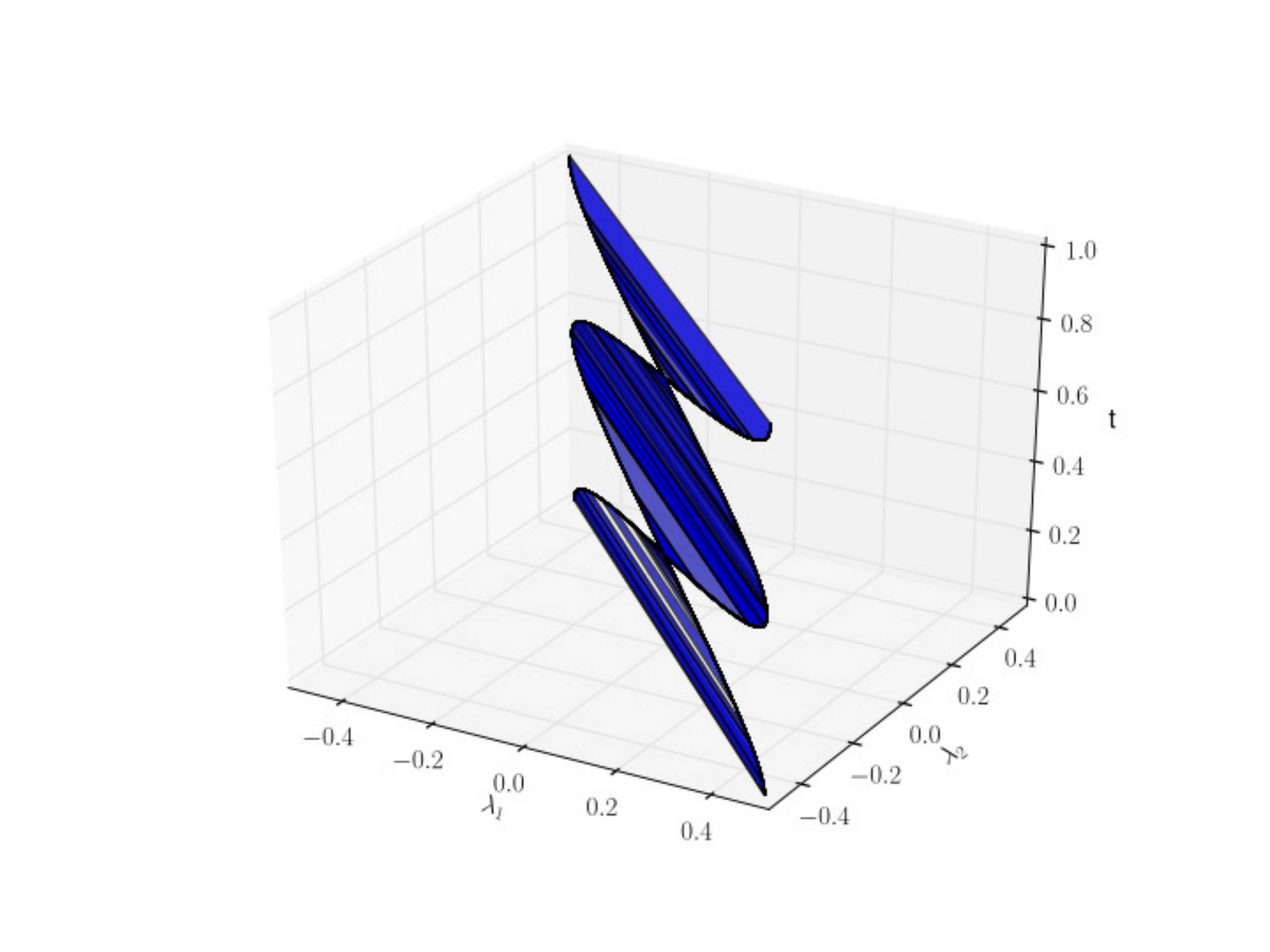}
  (b) \includegraphics[width=0.47\textwidth]{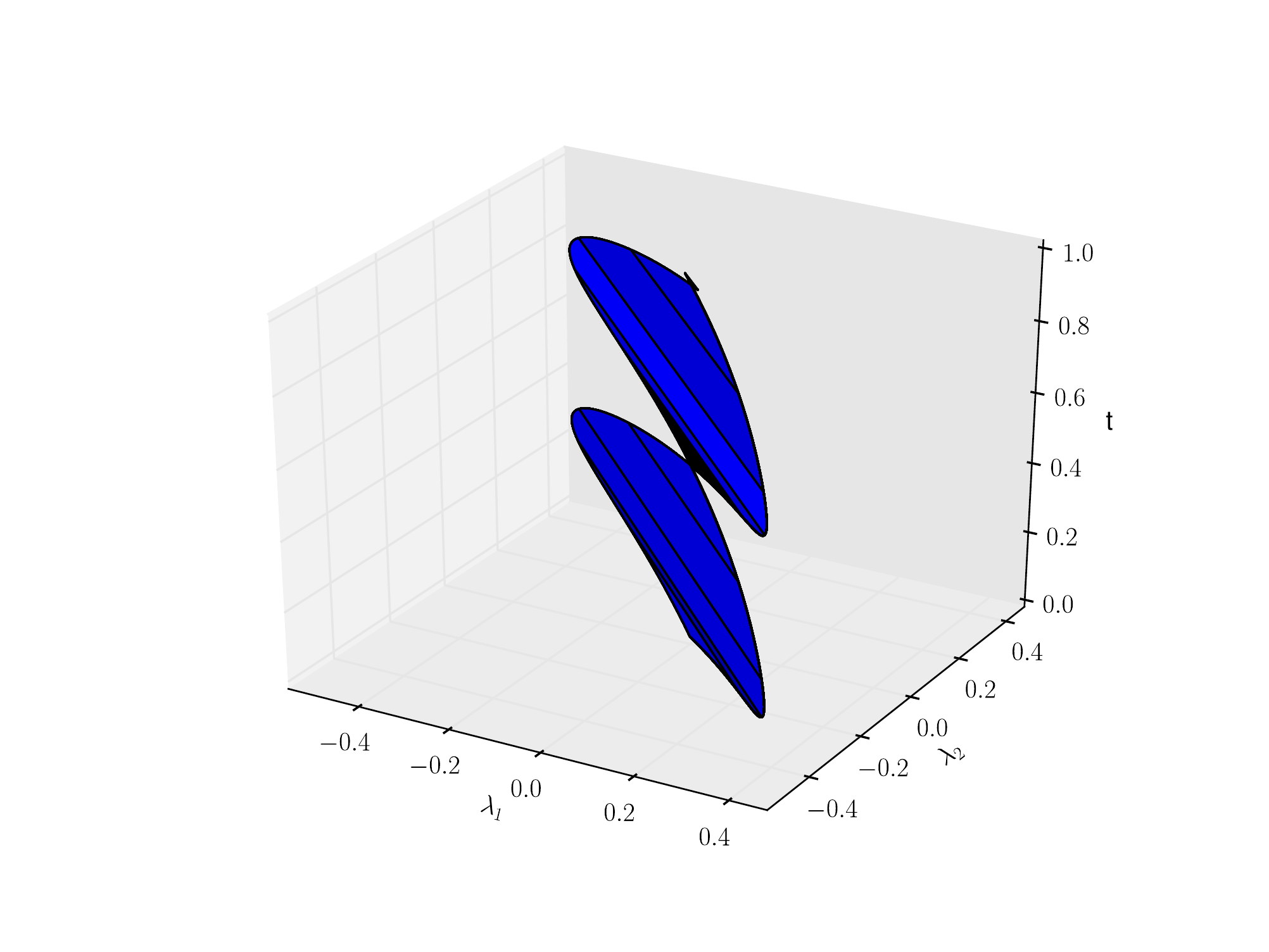}
}

\caption{\sf \label{fig_eigen2}
Isometric plots  of surfaces generated by the eigenvalues of the symmetric matrices $X({s},t)$ (Panel a) and $Y({s},t)$ (Panel b) under the  Bloch-Iserles strand equation for helical initial conditions. Time is shown vertically In all cases the linear oscillatory motion of period one dominates. Animations of the evolution of these solutions are available as supplementary materials for all numerical experiments presented in this paper.} 
\label{figBIlambda2}
\end{figure}

\newpage

\section{The $G$-strand PDE for $G={\rm Diff}(\mathbb{R})$}\label{DiffStrand-sec}
A particularly interesting $G$-strand system arises when we choose $G={\rm Diff}(\mathbb{R})$ and the Lagrangian involves the $H^1$ Sobolev norm.  This case is reminiscent of fluid dynamics and may be written naturally in terms of right-invariant tangent vectors $(t,s,x)$ and $\gamma(t,s,x)$ defined by 
\begin{equation}
\partial_t{g}= \nu \circ g
\quad\hbox{and}\quad
\partial_s{g}= \gamma \circ g
\,,
\end{equation}
where the symbol $\circ$ denotes composition of functions. In this right-invariant case, the $G$-strand PDE system in equations (\ref{aux-eqn-2time}) and (\ref{2timeEP}) for $G={\rm Diff}(\mathbb{R})$ with reduced Lagrangian $\ell(\nu,\gamma)$ takes the following form, which generalizes to $G={\rm Diff}(\mathbb{R}^d)$ in any number of spatial dimensions,
\begin{align}
\begin{split}
\frac{\partial}{\partial t} \frac{\delta\ell}{\delta \nu}
+ \frac{\partial}{\partial s}  \frac{\delta\ell}{\delta \gamma }
&= -\, {\rm ad}^*_{\nu}\frac{\delta\ell}{\delta \nu}
- {\rm ad}^*_{\gamma }\frac{\delta\ell}{\delta \gamma }
\,,
\\ 
\frac{\partial \gamma}{\partial t}   - \frac{\partial \nu}{\partial s} 
&=  {\rm ad}_\nu \gamma 
\,.
\end{split}
\label{Gstrand-eqn1R}
\end{align}
These equations form a subset of the equations studied in \cite{GBRa2008,Ho2002,HoKu1988,Tr2012} for complex fluids. They also form a subset of the equations for molecular strands studied in \cite{ElGBHoPuRa2010}. The latter comparison further justifies the name \emph{$G$-strands} for the systems being studied here. These equations may be derived from Hamilton's principle for an affine Lie group action, under which the auxiliary equation for $\gamma$ may be interpreted as an advection law. This interpretation is discussed further in \cite{ElGBHoPuRa2010,GBRa2008}.

Upon setting $m={\delta\ell}/{\delta \nu }$ and $n={\delta\ell}/{\delta \gamma }$, the right-invariant $G$-strand equations in (\ref{Gstrand-eqn1R}) for maps $\mathbb{R}\times\mathbb{R}\to G={\rm Diff}(\mathbb{R})$ in one spatial dimension may be expressed as a system of two 1+2 PDEs in $(t,s,x)$, 
\begin{align}
\begin{split}
m_t + n_s
&= 
-\, {\rm ad}^*_{\nu }m
- {\rm ad}^*_{\gamma }n
=
- (\nu m)_x - m\nu_x - (\gamma n)_x -  n\gamma_x
\,,
\\ 
\gamma_t - \nu _s
&= -\,{\rm ad}_\gamma \nu
=
-\nu \gamma_x + \gamma \nu_x 
\,.
\end{split}
\label{Gstrand-eqn2R}
\end{align}
The corresponding Hamiltonian structure for these  ${\rm Diff}(\mathbb{R})$-strand equations is obtained by Legendre transforming to 
\[
h(m,\gamma)=\langle m,\, \nu\rangle - \ell(\nu,\,\gamma)
\,.\]
One may then write the $m$-$\gamma$ equations (\ref{Gstrand-eqn2R}) 
in Lie-Poisson Hamiltonian form as
\begin{equation}
\frac{d}{dt}
\begin{bmatrix}
m \\ \gamma
\end{bmatrix}
=
\begin{bmatrix}
-\,{\rm ad}^*_\Box m  & \partial_s + {\rm ad}^*_\gamma
\\
\partial_s - {\rm ad}_\gamma  & 0
\end{bmatrix}
\begin{bmatrix}
{\delta h}/{\delta m} = \nu
\\
{\delta h}/{\delta \gamma} = -\, n
\end{bmatrix}.
\label{1stHamForm}
\end{equation}
This is the Lie-Poisson bracket dual to the action of the semidirect-product Lie algebra
\[
\mathfrak{g} 
= \mathfrak{X}(\mathbb{R})\,\circledS\, \Lambda^1({\rm Dens})(\mathbb{R})
\oplus C(\partial_s)
\]
in which $\mathfrak{X}(\mathbb{R})$ is the space of vector fields and $\Lambda^1({\rm Dens})(\mathbb{R})$  is the space of 1-form densities on the real line $\mathbb{R}$, plus a generalized 2-cocycle $C(\partial_s)$. In the $\gamma-\nu$ entry of the Hamiltonian matrix in (\ref{1stHamForm}), one recognizes the vector-field covariant derivative in $s$, and finds its negative adjoint operator in the $\nu-\gamma$ entry. For more details about how such Lie-Poisson Hamiltonian structures arise in complex fluids with finite-dimensional order-parameter groups (broken symmetries) and full discussions of their Lie algebraic properties, see \cite{GBRa2008,Ho2002,HoKu1988,Tr2012} and references therein. The differences in sign between the Hamiltonian matrix in (\ref{1stHamForm}) and those in (\ref{LP-Ham-struct-symbol}) and (\ref{LP-Ham-struct-sp2}) are due to the differences between right-invariance here and left-invariance in the other two cases.
\subsection{Semi-stationary solutions}
\begin{remark}
For $s$-independent solutions, the ${\rm Diff}(\mathbb{R})$ equations  (\ref{Gstrand-eqn2R}) reduce to EP equations for a Lagrangian defined on $T({\rm Diff}(\mathbb{R})\times \mathfrak{X}(\mathbb{R}))/{\rm Diff}_{\mathfrak{X}_0}(\mathbb{R}))$, in which ${\rm Diff}_{\mathfrak{X}_0}(\mathbb{R})\subset{\rm Diff}(\mathbb{R})$ is the isotropy subgroup of the vector field parameter $\mathfrak{X}_0(\mathbb{R})$. Namely, for solutions that are functions of $t$ and $x$,
\begin{align}
\begin{split}
m_t
&= 
-\, {\rm ad}^*_{\nu }m
+ {\rm ad}^*_{\gamma }(-n)
=
- (\nu m)_x - m\nu_x - (\gamma n)_x -  n\gamma_x
\,,
\\ 
\gamma_t
&= -\, {\rm ad}_\gamma \nu  
=
-\nu \gamma_x + \gamma \nu_x 
\,.
\end{split}
\label{m-gamma-eqns}
\end{align}
The EP equations for the reduced Lagrangian $\ell(\nu,\gamma)$ represent ideal fluid dynamics with an advected vector field, $\gamma$. It's Lie-Poisson form is given by the Hamiltonian matrix in (\ref{1stHamForm}), without the generalized 2-cocycle $C(\partial_s)$.

The semi-stationary equations (\ref{m-gamma-eqns}) may be interpreted as a one-dimensional analogue of ideal incompressible magnetohydrodynamics (iiMHD) in three dimensions \cite{ArKh1998,HoMaRa1998}. 
To compare the semi-stationary versions of (\ref{Gstrand-eqn2R}) with the iiMHD EP equations, we recall the latter in three dimensions for fluid velocity $u\in\Lambda^1(\mathbb{R}^3)$ and magnetic field $B\in\mathfrak{X}(\mathbb{R}^3)$, which may be rewritten in the present geometric notation as 
\begin{align}
\begin{split}
u_t &= -\, {\rm ad}^*_{u }u + {\rm ad}^*_{B }B 
\,,\\
B_t &= - {\rm ad}_{B }u
\,,
\end{split}
\label{mhd-raw}
\end{align}
with ${\rm div}\,u=0$ and ${\rm div}\,B=0$.
One sees the similarity with (\ref{m-gamma-eqns}) immediately.
These iiMHD equations arise as EP equations from Hamilton's principle with Lagrangian \cite{HoMaRa1998}
\begin{align}
\ell (u,B) = \tfrac12 \int \left(u^2 - B^2\right) d^3x 
= \frac12\, \| u \|^2_{L^2} - \frac12\, \| B  \|^2_{L^2}
\label{mhd-Lag}
\,.\end{align}
Equations (\ref{mhd-raw}) may be written in Lie-Poisson Hamiltonian form (\ref{1stHamForm}) as
\begin{equation}
\frac{d}{dt}
\begin{bmatrix}
u \\ B
\end{bmatrix}
=
\begin{bmatrix}
-\,{\rm ad}^*_\Box u  &  {\rm ad}^*_B
\\
-\, {\rm ad}_B  & 0
\end{bmatrix}
\begin{bmatrix}
{\delta h}/{\delta u} = u
\\
{\delta h}/{\delta B} = B
\end{bmatrix},
\label{mhdHamForm}
\end{equation}
with MHD Hamiltonian $h (u,B)$ given by the sum of the kinetic and magnetic energies,
\begin{equation}
h (u,B) = \tfrac12 \int \left(u^2 + B^2\right) d^3x 
\,.
\label{mhdHam}
\end{equation}
For further discussion of the EP theory underlying the equations for ideal fluid dynamics with advected quantities, see \cite{HoMaRa1998}.
\end{remark}

\subsection{Relation to the Camassa-Holm equation}
An {interesting subcase} of the system of semi-stationary ${\rm Diff}(\mathbb{R})$-strand equations (\ref{m-gamma-eqns}) arises when one selects the Lagrangian expressed only in terms of $\nu$ as its $H^1$ norm on the real line,
\begin{equation}
\ell (\nu ,\gamma ) = \frac12 \|\nu \|^2_{H^1} 
\,,
\label{CH-pkn-Lag}
\end{equation}
with vanishing boundary conditions, as $|x|\to\infty$. In this case, $m = \nu - \nu_{xx}$, 
so the semi-stationary $G$-strand equations in (\ref{m-gamma-eqns}) provide an extension of the completely integrable Camassa-Holm (CH) equation  \cite{CaHo1993},
\begin{equation}
m_t 
= 
-\, {\rm ad}^*_{\nu }m
=
- (\nu m)_x - m\nu_x 
\quad\hbox{with}\quad
m = \frac{\delta\ell}{\delta \nu } = \nu - \nu_{xx}
\,.
\label{CH-eqn}
\end{equation}
Specifically, these modified $G$-strand equations reduce in the absence of  $\gamma$-dependence to the CH equation, which admits singular solutions known as \emph{peakons} \cite{CaHo1993}
\begin{align}
m(t,x) = \sum_a M_a(t)\delta(x-Q^a(t))
\,,
\label{pkn-soln}
\end{align}
where we sum in $a\in \mathbb{Z}$ over the integers, or over any subset of the integers. 
The peakon solution (\ref{pkn-soln}) of the CH equation may be understood as a singular momentum map, as discussed in \cite{HoMa2004}.

\subsection{Peakon solutions of the ${\rm Diff}(\mathbb{R})$-strand  equations}
 With the choice of Lagrangian using the $H^1$ norm, 
 \begin{equation}
\ell (\nu ,\gamma ) = \frac12 \|\nu \|^2_{H^1} - \frac12 \|\gamma \|^2_{H^1} 
\,,
\label{Gstrand-pkn-Lag}
\end{equation}
 the ${\rm Diff}(\mathbb{R})$-strand equations (\ref{Gstrand-eqn2R}) admit peakon solutions in both momenta $m$ and $n$, with continuous velocities $\nu$ and $\gamma$. We state this result in the following theorem.
\begin{theorem}\label{Gstrand-HP}
The ${\rm Diff}(\mathbb{R})$-strand equations (\ref{Gstrand-eqn2R}) 
admit singular solutions 
\begin{align}
m(s,t,x) &= \sum_a M_a(s,t)\delta(x-Q^a(s,t))
\,,\qquad
n(s,t,x) = \sum_a N_a(s,t)\delta(x-Q^a(s,t))
\,,
\label{Gstrand-singsolns}
\\
\nu(s,t,x) & =K*m=\sum_a M_a(s,t) K(x,Q^a)
\,,\qquad
\gamma(s,t,x)  = K*n=\sum_a N_a(s,t) K(x,Q^a)
\,,
\nonumber
\end{align}
that are peakons for $K(x,y)= \frac12  e^{-|x-y|}$.
These singular solutions follow from Hamilton's principle $\delta S=0$ for the constrained action $S=\int L(\nu,\gamma,Q)\,dt$ given by
\begin{eqnarray*}
S = \int \ell(\nu,\gamma) 
+ \sum_a M_a(s,t) \big( \partial_t Q^a(s,t) - \nu(Q^a,s,t) \big)  
+ \sum_a   N_a(s,t) \big( \partial_s Q^a(s,t) - \gamma(Q^a,s,t) \big)
ds\,dt
\,.
\end{eqnarray*}
and they are \emph{peakons} for the Lagrangian 
$\ell(\nu,\gamma) $ given in equation (\ref{Gstrand-pkn-Lag}).
\end{theorem}
\begin{proof}
After replacing 
\[
\nu(Q^a,s,t) = \int \nu(x,s,t)\delta(x-Q^a(t,s))\, dx
\quad\hbox{and}\quad
\gamma(Q^a,s,t) = \int \gamma(x,s,t)\delta(x-Q^a(t,s))\, dx
\,,\]
one computes $\delta S$ in Hamilton's principle $\delta S =0$ for the Lagrangian $\ell(\nu,\gamma) $ as,
\begin{align*}
\delta S &= \int 
\left\langle \frac{\delta\ell}{\delta \nu} 
- \sum_a M_a(t,s)\delta(x-Q^a(t,s)),  \delta \nu \right\rangle
+
\left\langle \frac{\delta\ell}{\delta \gamma} 
- \sum_a N_a(t,s)\delta(x-Q^a(t,s))
\,, \delta \gamma \right\rangle
ds\,dt
\\&\quad
+ \int
  \sum_a  \delta M_a (s,t) \left(  \partial_t Q^a(s,t) - \nu(Q^a,s,t) \right)  
+ \sum_a\delta N_a(s,t) \left( \partial_s Q^a(s,t) - \gamma(Q^a,s,t) \right)
ds\,dt
\\&\quad
- 
\int \sum_a \left( \partial_t M_a + \partial_s N_a
+ \sum_b\left(
\frac{\partial \nu(Q^b)}{\partial Q^a} M_b  
+ \frac{\partial \gamma(Q^b)}{\partial Q^a} N_b \right)
\right)\delta Q^a \, ds\,dt
\,.\end{align*}
Hence, stationarity of the constrained action $S$ in the statement of the theorem implies
\begin{align}
\begin{split}
\delta \nu:&\quad  \frac{\delta\ell}{\delta \nu} 
= \sum_a M_a(t,s)\delta(x-Q^a(t,s))
\,,\\
\delta \gamma:&\quad \frac{\delta\ell}{\delta \gamma} 
= \sum_a N_a(t,s)\delta(x-Q^a(t,s))
\,,\\
\delta M_a: &\quad \partial_t Q^a(s,t) = \nu(Q^a,s,t)
\,,\\[2mm]
\delta N_a: &\quad  \partial_s Q^a(s,t) = \gamma(Q^a,s,t)
\,,\\
\delta Q^a:&\quad 
\partial_t M_a + \partial_s N_a
+ \sum_b \left(\frac{\partial \nu(Q^b)}{\partial Q^a} M_b
 + \frac{\partial \gamma(Q^b)}{\partial Q^a} N_b\right) =0
\,.
\end{split}
\label{momap}
\end{align}
The first two equations recover the forms of the singular solutions from equation (\ref{Gstrand-singsolns}) in the statement of the theorem. These expressions are \emph{singular momentum maps}, as discussed in detail in \cite{HoMa2004}.
The corresponding velocities are expressed as
\begin{align}
\nu(x,s,t)=K*m=\sum_b M_b(s,t) K(x,Q^b)
\,,\quad
\gamma(x,s,t)=K*n=\sum_b N_b(s,t) K(x,Q^b)
\,,
\label{nu-gamma-K-eqn} 
\end{align}
for a  symmetric, positive-definite kernel $K(x,y)$.

Inserting the forms of the solutions for the tangent vectors, or velocities $\nu$ and $\gamma$ in (\ref{nu-gamma-K-eqn}) into the compatibility condition in (\ref{Gstrand-eqn2R}) yields
\begin{align}
\begin{split}
\sum_b K(x,Q^b)
&
 \left[ -\partial_t N_b(s,t) + \partial_s M_b
+    \sum_c (N_bM_c - M_bN_c) \frac{\partial K(x,Q^c)}{\partial
x}\right]
\\&
+\sum_{b,c} K^{bc}\frac{\partial K(x,Q^b)}{\partial
x} (N_bM_c - M_bN_c)=0
\,.
\end{split}
\label{nu-gamma-x-eqn}
\end{align}
Evaluating this formula at $x=Q^a(s,t)$ then yields a relation among the solution parameters as
\begin{align}
&\sum_b\left[ K^{ab}(\partial_s M_b   -  \partial_t N_b)
+    \sum_c (N_bM_c - M_bN_c) \frac{\partial K^{ac}}{\partial Q^a}
(K^{ab} -K^{cb})\right] = 0
\,,
\label{MNK-sumeqn}
\end{align}
where we have introduced the matrix notation $K^{ab}:=K(Q^a,Q^b)$.
The symmetric matrix $K^{ab}=K^{ba}$ follows from the choice of kernel and is assumed to be invertible. Invertibility of $K^{ab}$ in the compatibility equation (\ref{MNK-sumeqn}) implies a relation 
\begin{align}
\partial_s M_e   -  \partial_t N_e
+    \sum_{a,b,c} (N_bM_c - M_bN_c) \frac{\partial K^{ac}}{\partial Q^a} (K^{ab}-K^{cb})(K^{-1})_{ea}
= 0
\,,
\label{MNK-eqn}
\end{align}
which may also be regarded as an evolution equation for the Lagrange multiplier $N_e(s,t)$.

Equation (\ref{MNK-sumeqn}), or equivalently equation (\ref{MNK-eqn}) also arises from applying equality of cross derivatives $\partial^2_{st}Q=\partial^2_{ts}Q$ to the two constraint equations in (\ref{momap}) imposed by the Lagrange multipliers $M_a$ and $N_a$. This may be seen by substituting the two constraint relations into the equations for the cross derivatives, 
\begin{align}
\partial^2_{ts}Q^a = \partial_s \nu(Q^a,s,t) +  \sum_c \frac{\partial \nu (Q^a)}{\partial Q^c}\gamma (Q^c)
= \partial_t \gamma(Q^a,s,t) +  \sum_c \frac{\partial \gamma(Q^a)}{\partial Q^c}\nu(Q^c)
= \partial^2_{st}Q
\,,
\label{cross-deriv}
\end{align}
whereupon rearranging then yields equation (\ref{MNK-sumeqn}) and invertibility of $K^{ab}$ implies (\ref{MNK-eqn}).

Finally, the $\delta Q^a$ equation in the set (\ref{momap}) is given by, \begin{align}
&
\partial_t M_a + M_a\sum_b M_b \frac{\partial K^{ab}}{\partial Q^a} 
+ \partial_s N_a + N_a \sum_b N_b \frac{\partial K^{ab}}{\partial Q^a} =0
\,,
\label{dQ-eqn}
\end{align}
which is an evolution equation for the Lagrange multiplier $M_a(s,t)$.

For the Lagrangian in (\ref{Gstrand-pkn-Lag}) we have $\nu=K*m$ and $\gamma=K*n$ for the symmetric positive definite kernel $K(x,y)=\frac12e^{-|x-y|}$, which means that in this case matrix elements of $K^{ab}$ take the peakon shape $K^{ab}=\frac12e^{-|Q^a-Q^b|}$.  The corresponding singular solutions (\ref{Gstrand-singsolns}) are peakons for the Lagrangian in (\ref{Gstrand-pkn-Lag}). 
\end{proof}

\paragraph{Summary.}
Upon collecting equations, we have the following three results. 
\begin{enumerate}
\item
The singular solutions of the $G$-strand equations for $G={\rm Diff}(\mathbb{R})$ are represented by two singular momentum maps,
\begin{align}
\begin{split}
m(s,t,x) &= \sum_a M_a(s,t)\delta(x-Q^a(s,t))
\,,\quad
T_t^*{\rm Emb}(\mathbb{S},\mathbb{R})\to \mathfrak{X}_t^*(\mathbb{R})
\,,\\
n(s,t,x) &= \sum_a N_a(s,t)\delta(x-Q^a(s,t))
\,,\quad
T_s^*{\rm Emb}(\mathbb{S},\mathbb{R})\to \mathfrak{X}_s^*(\mathbb{R})
\,,
\end{split}
\label{Gstrand-momap}
\end{align}
where $\mathbb{S}\in\mathbb{R}$ is the support set of the delta functions, and ${\rm Emb}(\mathbb{S},\mathbb{R})$ denotes the set of smooth
embeddings $Q : \mathbb{S}\to\mathbb{R}$. 
\item
The two tangent vectors, or velocities, that correspond to these momentum maps are, 
\begin{align}
\begin{split}
\nu(t,s,x) & =K*m=\sum_a M_a(s,t) K(x,Q^a(s,t))
\,,\\
\gamma(t,s,x) & = K*n=\sum_a N_a(s,t) K(x,Q^a(s,t))
\,.
\end{split}
\label{Gstrand-vel}
\end{align}
These arise from the convolution $K*: \mathfrak{X}^*(\mathbb{R})\to\mathfrak{X}(\mathbb{R})$ of $m$ and $n$ with the kernel $K$.
\item
The solution parameters $\{Q^a(s,t), M_a(s,t), N_a(s,t)\}$ with $a\in\mathbb{Z}$ that specify the singular solutions (\ref{Gstrand-singsolns}) are determined by the following set of evolutionary PDEs in $s$ and $t$, in which we denote $K^{ab}:=K(Q^a,Q^b)$:
\begin{align}
\begin{split}
\partial_t Q^a(s,t) &= \nu(Q^a,s,t) = \sum_b M_b(s,t) K^{ab}
\,,\\
\partial_s Q^a(s,t) &= \gamma(Q^a,s,t) =\sum_b N_b(s,t) K^{ab}
\,,\\
\partial_t M_a(s,t) &= -\, \partial_s N_a
-\sum_c (M_aM_c+N_aN_c) \frac{\partial K^{ac}}{\partial Q^a} 
\quad\hbox{(no sum on $a$),}
\\
\partial_t N_e(s,t) &=\partial_s M_e  
+    \sum_{a,b,c} (N_bM_c - M_bN_c) \frac{\partial K^{ac}}{\partial Q^a} (K^{ab}-K^{cb})(K^{-1})_{ea}
\quad\hbox{(do sum on $a$).}
\end{split}
\label{Gstrand-eqns}
\end{align}
\end{enumerate}
The last pair of equations in (\ref{Gstrand-eqns}) may be solved as a system for the momenta, or Lagrange multipliers $(M_a,N_e)$, then used in the previous pair to update the support set of positions $Q^a(t,s)$. Given $Q^a(t,s)$ for $a\in\mathbb{Z}$, one constructs $(m,n(t,s,x))$ along the solution paths $x=Q^a(t,s)$ from equations (\ref{Gstrand-momap}) and then obtains $(\nu,\gamma(t,s,x))$ for $x\in\mathbb{R}$ from equations (\ref{Gstrand-vel}). Alternatively, knowing the position $Q^a(s,t)$, $a=1,\dots,n$, for all $s$ at a given time $t$, also determines $N_a$ upon inverting the matrix $K^{ab}$ in the second equation in (\ref{Gstrand-eqns}).


\subsection{Canonical Hamiltonian form of the ${\rm Diff}(\mathbb{R})$-strand peakon dynamics}

The Diff$(\mathbb{R})$-strand peakon equations (\ref{Gstrand-eqns}) may be written in a \emph{canonical} Hamiltonian form, after making a Legendre transformation to a \emph{Routhian}, or constrained Hamiltonian, by forming
\begin{align}
\begin{split}
H(M,Q) &=  \int \Big[ \sum_a M_a\partial_t Q^a\Big] ds - L(\nu,\gamma,Q)
\\
&= \int \Big[\int \sum_a M_a\nu(x) \delta(x-Q^a) dx - \ell(\nu,\gamma) 
- \sum_a N_a(\partial_s Q^a - \gamma(Q^a))\Big]ds
\,,\end{split}
\end{align}
where $L(\nu,\gamma,Q)$ denotes the constrained Lagrangian appearing  in Theorem \ref{Gstrand-HP},
\begin{align}
\begin{split}
L(\nu,\gamma,Q) &:=  \int \Big[ \ell(\nu,\gamma) 
+ \sum_a
\left\langle  M_a(s,t), \partial_t Q^a(s,t) - \nu(Q^a,s,t) \right\rangle
\\ &\hspace{1in}
+ \sum_a\left\langle  N_a(s,t), \partial_s Q^a(s,t) - \gamma(Q^a,s,t) \right\rangle\Big]
ds
\,.
\end{split}
\end{align}
The variations of the resulting Routhian $H$ are given by
\begin{align}
\begin{split}
\delta H(M,Q) &= \int \bigg[
\bigg\langle
\sum_a M_a\delta(x-Q^a) - \frac{\delta \ell}{\delta \nu}, \delta \nu 
\bigg\rangle
+
\bigg\langle
\sum_a N_a\delta(x-Q^a) - \frac{\delta \ell}{\delta \gamma}, \delta \gamma \bigg\rangle
\\[2mm]
&\qquad 
- \sum_a \big(\partial_s Q^a - \gamma(Q^a)\big)\delta N_a
\\
&\qquad 
+ \sum_a \nu(Q^a) \delta M_a
+ \sum_{a,b} \bigg( M_b \frac{\partial \nu(Q^b)}{\partial Q^a}
 +
\partial_s N_a + N_b \frac{\partial \gamma(Q^b)}{\partial Q^a} \bigg) 
 \delta Q^a
 \bigg] ds
 \,,
\end{split}
\end{align}
where the angle brackets $\langle\,\cdot\,,\,\cdot\,\rangle$ denote $L^2$ pairing in $x\in\mathbb{R}$. Vanishing of the first two terms recovers the two singular moment maps in (\ref{Gstrand-momap}). Vanishing of the third term imposes the second constraint in (\ref{Gstrand-eqns}). Finally, the remaining two terms produce the corresponding \emph{canonical}  equations for the Routhian $H$ with canonically conjugate variables $Q^a(s,t)$ and $M_a(s,t)$,
\begin{align}
\begin{split}
\partial_t Q^a 
&= \frac{\delta H}{\delta M_a} = \nu(Q^a)
\,,\\
\partial_t M_a 
&= -\, \frac{\delta H}{\delta Q^a} 
=
-\,\partial_s N_a
-
\sum_{b} \bigg( M_b \frac{\partial \nu(Q^b)}{\partial Q^a}
 + N_b \frac{\partial \gamma(Q^b)}{\partial Q^a} \bigg) 
.
\end{split}
\end{align}
Compatibility of the constraint $\partial_s Q^a = \gamma(Q^a)$ and the canonical equation $\partial_t Q^a = \nu(Q^a)$ then recovers the last equation in (\ref{Gstrand-eqns}) for the evolution of the Lagrange multiplier $N_a$. This completes the canonical Hamiltonian formulation of the singular solutions of the Diff$(\mathbb{R})$-strand equations in (\ref{Gstrand-eqns}). The conserved energy of the peakons for the Lagrangian in (\ref{Gstrand-pkn-Lag}) is
\begin{align}
E= \frac12\int \sum_a M_a\nu(Q^a) + N_a\gamma(Q^a)\,ds
= \frac12\int \sum_{a,b} M_aK^{ab}M_b + N_aK^{ab}N_b\,ds
\,.
\label{Diff(R)-pkn-erg}
\end{align}

\subsection{Alternative forms of the ${\rm Diff}(\mathbb{R})$-strand equations}

\paragraph{Velocities, Eulerian.}
The ${\rm Diff}(\mathbb{R})$-strand equations (\ref{Gstrand-eqn2R}) may be written equivalently in terms of the velocities $\nu$ and $\gamma$, in Eulerian form as 
\begin{align}
\begin{split}
\partial_t \nu + \nu\nu_x + \partial_s \gamma + \gamma\gamma_x
&=
-\,\partial_x P(\nu,\gamma)
\,,\\[1mm]
\partial_t \gamma + \gamma\nu_x
&=
\partial_s \nu + \nu\gamma_x
\,,\\
\hbox{where}\quad 
P(\nu,\gamma) := K*&\Big( \nu^2 + \frac12 \nu_x^2
+ \gamma^2+\frac12 \gamma_x^2\Big).
\end{split}
\label{DiffRstrand-velocity-form}
\end{align}

\paragraph{Velocities, Lagrangian.}
Equations (\ref{DiffRstrand-velocity-form}) themselves may be rewritten equivalently by using Lagrangian time derivatives in $s$ and $t$, in the compact form,
\begin{align}
\begin{split}
\frac{d\nu}{dt}\bigg|_\nu + \frac{d\gamma}{ds}\bigg|_\gamma
&=
-\,\partial_x P(\nu,\gamma)
\,,\\
\frac{d\gamma}{dt}\bigg|_\nu - \frac{d\nu}{ds}\bigg|_\gamma
&=
0
\,.
\end{split}
\end{align}

\paragraph{Momenta, Lagrangian.}
Going back to the momentum 1-form densities $m\,dx^2$ and $n\,dx^2$ allows us to rewrite equations (\ref{Gstrand-eqn2R}) once more equally compactly using Lagrangian time derivatives as
\begin{align}
\begin{split}
\frac{d}{dt}\bigg|_{{dx}/{dt}=\nu}\hspace{-5mm}(m\,dx^2) 
+ \frac{d}{ds}\bigg|_{{dx}/{ds}=\gamma} \hspace{-5mm} (n\,dx^2) 
&=
0
\,,\\[2mm]
\frac{d}{dt}\bigg|_{{dx}/{dt}=\nu}\hspace{-5mm}\gamma\
-\
 \frac{d}{ds}\bigg|_{{dx}/{ds}=\gamma} \hspace{-5mm} \nu 
&= 0
\,.
\end{split}
\end{align}
\paragraph{Outlook:}
Studies of the solution behavior and the issue of integrability of the $G$-strand equations (\ref{Gstrand-eqn2R}) for $G={\rm Diff}(\mathbb{R})$ with the Lagrangian (\ref{Gstrand-pkn-Lag}) may be expected in future investigations along the present lines. For example, the velocity form of the ${\rm Diff}(\mathbb{R})$-strand equations in (\ref{DiffRstrand-velocity-form}) may be useful in determining whether they produce wave breaking from smooth initial conditions, by studying their slope dynamics at a moving inflection point of both velocities to prove a steepening lemma.

\paragraph{Higher spatial dimensions.}
The $G$-strand equations may also be written in higher spatial dimensions. For example, the $G$-strand equations for volume-preserving diffeomorphisms in three dimensions with $G={\rm SDiff}(\mathbb{R}^3)$ may be written in terms of divergence-free vector fields as 
\begin{align}
\begin{split}
\partial_t \boldsymbol{\omega} 
+ [\boldsymbol{\nu},\boldsymbol{\omega}] 
+ \partial_s \boldsymbol{\Omega} 
+ [\boldsymbol{\gamma},\boldsymbol{\Omega}] &= 0
\,,\quad\hbox{with}\quad 
\boldsymbol{\omega} = {\rm curl}\,\boldsymbol{\nu}
\quad\hbox{and}\quad 
\boldsymbol{\Omega} = {\rm curl}\,\boldsymbol{\gamma}
\,,\\
\partial_t \boldsymbol{\nu} - \partial_s \boldsymbol{\gamma} 
- [\boldsymbol{\nu},\,\boldsymbol{\gamma}] &=0
\,,\quad\hbox{where}\quad 
 {\rm div}\,\boldsymbol{\nu} = 0
 \quad\hbox{and}\quad 
 {\rm div}\,\boldsymbol{\gamma} = 0
\,.
\end{split}
\end{align}
It may be interesting to study the point-vortex dynamics of these equations in two dimensions.

\section{Conclusions}\label{conclusion-sec}

 
This paper has provided a general framework for studying the $G$-strand equations, the PDE system comprising the Euler-Poincar\'e (EP) variational equations (\ref{2timeEP}) for a G-invariant Lagrangian, coupled to the auxiliary \emph{zero-curvature} equation (\ref{aux-eqn-2time}). The latter has often been the departure point and main focus in other approaches, because it sets up the Lax-pair formulation of the system.
 
Once the G-invariant Lagrangian has been specified, the EP and zero-curvature system of $G$-strand equations (\ref{2timeEP}) and (\ref{aux-eqn-2time}) follow automatically.  In the present paper, the well-known $SO(3)$ chiral model PDE system has been extended by placing it into the $G$-strand framework and then deriving its Lie-Poisson Hamiltonian formulation on the dual Lie algebra $so(3)^*$. The corresponding Lie-Poisson bracket given in equation (\ref{LP-Ham-struct-symbol}) turns out to be the same as for a perfect complex fluid \cite{Ho2002}. It also appears in the Lie-Poisson brackets for Yang-Mills fluids \cite{GiHoKu1982,GiHoKu1983} and for spin glasses \cite{HoKu1988}. The $SO(3)$ chiral model itself emerges in the $G$-strand framework for the very simple choice of quadratic Hamiltonian given in equation (\ref{chiral-Ham}). 

In this framework, two new integrable 1+1 EP PDE systems have been identified in Section \ref{SO3K-sec} and Section \ref{SP2-sec}  that may be transformed into the $SO(3)$ P-chiral model.  These are the $SO(3)_K$ and $Sp(2)$ $G$-strand equations in which the corresponding Hamiltonians are sums of squares. 
Other types of integrable P-chiral models are also available. The various Lie algebras for these may be classified according to the signature ${\rm sgn}(\sf K)$ of the real symmetric matrix $\sf K$, defined as the triple $(l,m,n)$ of the numbers of its positive, null, and negative eigenvalues, respectively.
 
Finally, the dynamics of the soliton solutions for these two integrable systems has been studied numerically and their coherent nonlinear scattering behavior has been explored for some simple initial conditions in Section \ref{JRPnumerics}. 

Other systems may be treated using the same approach. For example, one may formulate the $G$-strand equations for the two-dimensional Toda lattice, in terms of an appropriate Lie-algebra of lower-triangular matrices, as reviewed in \cite{Gu2006}. The multicomponent Bloch-Iserles system may also be formulated in this way. 
These other systems are likely to be fruitful subjects of future investigations of the $G$-strand PDE, including the $G$-strands on the diffeomorphisms, introduced in Section \ref{DiffStrand-sec}.

\paragraph{Acknowledgments.} We thank our friends A. M. Bloch, C. J. Cotter, F. Gay-Balmaz, A. Iserles, T. S. Ratiu and C. Tronci for their kind encouragement and thoughtful remarks during the course of this work. 
We also thank the referees for their encouraging remarks and suggestions, particularly for encouraging us to pursue the development of  the $G$-strands on the diffeomorphisms, introduced in Section \ref{DiffStrand-sec}. 
DDH gratefully acknowledges partial support by the Royal Society of London's Wolfson Award scheme and the European Research Council's Advanced Grant. The research of DDH and RII in partially supported by the EU Seventh Framework programme FP7/2007-2013 under grant agreement No 235419. JRP thanks the US Office of Naval Research for support.


\end{document}